%% file: main.tex
\documentclass[runningheads]{llncs}

\usepackage{graphicx}
\usepackage{todonotes}
\usepackage{subfiles} 
\usepackage{amsmath}
\usepackage{amsfonts}
\usepackage{mathtools}
\usepackage{wrapfig}
\usepackage{multicol}
\usepackage{thm-restate}
\usepackage{algpseudocode}

\usepackage{algorithm}
\algtext*{EndFor}

\newcounter{claimcounter}
\numberwithin{claimcounter}{theorem}

\input{macro}

\begin{document}

\title{On eventual non-negativity and positivity for the weighted sum of powers of matrices \thanks{This work was partly supported by DST/CEFIPRA/INRIA Project EQuaVE, DST/SERB Matrices Grant MTR/2018/000744, and MHRD/IMPRINT-1/Project5496(FMSAFE).} 
}

\titlerunning{On eventual properties for weighted sum of powers of matrices}

\author{S Akshay \and Supratik Chakraborty \and Debtanu Pal}

\authorrunning{S. Akshay, Supratik Chakraborty and Debtanu Pal}

\institute{Indian Institute of Technology Bombay, Mumbai - 400076, India}

\maketitle              

\begin{abstract}
 The long run behaviour of linear dynamical systems is often studied by looking at eventual properties of matrices and recurrences that underlie the system. A basic problem that lies at the core of many questions in this setting is the following: given a set of pairs of rational weights and matrices $\{(w_1, A_1), \ldots, (w_m, A_m)\}$, we ask if the weighted sum of powers of these matrices is eventually non-negative (resp. positive), i.e.,  does  there exist an integer $N$ s.t for all $n \geq N$, $\sum_{i=1}^m w_i\cdot A_i^n  \geq 0$ (resp. $> 0$). The restricted setting when $m=w_1=1$, results in so-called eventually non-negative (or eventually positive) matrices, which enjoy nice spectral properties and have been well-studied in control theory. More applications arise in varied contexts, ranging from program verification to partially observable and multi-modal systems.

Our goal is to investigate this problem and its link to linear recurrence sequences. Our first result is that for $m\geq 2$, the problem is as hard as the ultimate positivity of linear recurrences, a long standing open question (known to be  $\mathsf{coNP}$-hard). Our second result is a reduction in the other direction showing that for any $m\geq 1$, the problem reduces to ultimate positivity of linear recurrences. This shows precise upper bounds for several subclasses of matrices by exploiting known results on linear recurrence sequences. Our third main result is a novel reduction technique for a large class of problems (including those mentioned above) over rational diagonalizable matrices to the corresponding problem over simple real-algebraic matrices. This yields effective decision procedures for diagonalizable matrices.
\end{abstract}

\section{Introduction}
\label{sec:intro}

The study of eventual or asymptotic properties of discrete-time linear dynamical systems  has long been of interest to both theoreticians and
practitioners. Questions pertaining to (un)-decidability and/or computational complexity of predicting the long-term behaviour of such systems have been extensively studied over the last few decades.  Despite significant advances, however, there remain simple-to-state questions that have eluded answers so far.  In this work, we investigate one such problem, explore its significance and links with other known problems, and study its complexity and computability landscape.  

The time-evolution of linear dynamical systems is often modeled using linear recurrence sequences, or using sequences of powers of matrices.  Asymptotic  properties of powers of matrices  are therefore of central interest in the study of linear differential systems, dynamic control theory, analysis of linear loop programs etc. (see e.g. \cite{10.1137/070693850}, \cite{DBLP:journals/siglog/OuaknineW15}, \cite{tiwari2004termination}, \cite{zaslavsky1991eventually}).
The
literature contains a rich body of work on the decidability and/or computational complexity of problems related to the long-term behaviour of such
systems (see, e.g. \cite{DBLP:conf/fct/Ouaknine13}, \cite{zaslavsky1991eventually}, \cite{10.1007/978-3-642-33512-9_3}, \cite{halava2005skolem}, \cite{DBLP:conf/hybrid/FijalkowOPP019}, \cite{tiwari2004termination}). 

A question of significant interest in this context is whether the powers of a given matrix of rational numbers eventually have only non-negative (resp. positive) entries. Such matrices, also called \emph{eventually non-negative} (resp. \emph{eventually  positive}) matrices, enjoy beautiful algebraic properties (\cite{NOUTSOS2006132}, \cite{CARNOCHANNAQVI2004245}, \cite{article}, \cite{friedland1978inverse}), and have been studied by mathematicians, control theorists and computer scientists, among others. 

For example, the work of~\cite{10.1137/070693850} investigates reachability and holdability of non-negative states for linear differential systems -- a problem in which eventually non-negative matrices play a central role.  Similarly,  eventual non-negativity (or positivity) of a matrix modeling a linear dynamical system makes it possible to apply the elegant Perron-Frobenius theory~\cite{rump2003perron}, \cite{maccluer2000many}  to analyze the long-term behaviour of the system beyond an initial number of time steps. Another level of complexity is added if the dynamics is controlled by a set of matrices rather than a single one. For instance, each matrix may model a mode of the linear dynamical system~\cite{lebacque2009cross}. In a partial observation setting(~\cite{lale2020logarithmic}, \cite{zhang2019learning}), we may not know which mode the system has been started in, and hence have to reason about eventual properties of this multi-modal system. It turns out that this problem reduces to analyzing the sum of powers of the per-mode matrices (see Section~\ref{sec:motiv}).

Motivated by the above considerations, we study the problem of determining whether a given matrix of rationals is eventually non-negative or eventually positive and also
a generalized version of this problem, wherein we ask if the \emph{weighted sum of powers of a given set of matrices of rationals} is eventually non-negative (resp. positive). Let us formalize the general problem statement.  \emph{\bfseries Given a set $\AAA = \{(w_1, A_1), \ldots (w_m, A_m)\}$, where each $w_i$ is a rational number and each $A_i$ is a $k\times k$  matrix of rationals, we wish to determine if $\sum_{i=1}^m w_i\cdot A_i^n$ has only non-negative  (resp. positive) entries for all sufficiently large values of $n$.} We call this problem \emph{Eventually Non-Negative (resp. Positive) Weighted Sum of Matrix Powers} problem, or {\ENSum} (resp. {\EPSum}) for short. The eventual non-negativity (resp. positivity) of powers of a single matrix is a special case of the above problem, where $\AAA = \{(1, A)\}$.  We call this special case the \emph{Eventually Non-Negative (resp. Positive) Matrix} problem, or {\ENId} (resp. {\EPosId}) for short. 

Given the simplicity of the {\ENSum} and {\EPSum} problem statements, one may be tempted to think that there ought to be simple algebraic characterizations that tell us whether
$\sum_{i=1}^m w_i\cdot A_i^n$ is eventually non-negative or positive. But in fact, the landscape is significantly nuanced. On one hand, a solution to the general {\ENSum} or {\EPSum} problem would resolve long-standing open questions in mathematics and computer science. On the other hand,  efficient algorithms can indeed be obtained under certain well-motivated conditions.

This paper is a study of both these aspects of the problem. Our primary contributions can be summarized as follows. Below, we use $\AAA = \{(w_1, A_1), \ldots (w_m, A_m)\}$ to define an instance of {\ENSum} or {\EPSum}.
\begin{enumerate}
\item If $|\AAA| \ge 2$, we show that both {\ENSum} and {\EPSum} are
  at least as hard as the ultimate non-negativity problem for
  linear recurrence sequences, which we call {\UPos} for short.  The decidability
  of {\UPos} is closely related to Diophantine approximations, and remains unresolved despite extensive research (see e.g.~\cite{ouaknine2017ultimate}).
 
  Since {\UPos} is {\mycoNP}-hard (in fact, as hard as the decision problem for universal theory of reals), so is  {\ENSum} and {\EPSum}, when $|\AAA| \ge 2$. Thus, unless {\myP} = {\myNP}, we cannot hope for polynomial-time algorithms, and any algorithm would also resolve long-standing open problems.
  \item On the other hand, regardless of $|\AAA|$, we show a reduction in the other direction from {\ENSum} (resp. {\EPSum}) to {\UPos} (resp. {\SUPos}, the strict version of {\UPos}). As a consequence, we get decidability and complexity bounds for special cases of {\ENSum} and {\EPSum}, by exploiting recent results on recurrence sequences
  \cite{ouaknine2017ultimate}, \cite{DBLP:journals/corr/OuaknineW13},
  \cite{s_et_al:LIPIcs:2017:8130}. For example, if each matrix $A_i$ in $\AAA$ is simple, i.e. has all distinct eigenvalues, we obtain {\myPSPACE} algorithms.
  \item Finally, we consider the case where $A_i$ is diagonalizable (also called non-defective or inhomogenous dilation map) for each $(w_i,A_i)\in \AAA$. This is a practically useful class of matrices and strictly subsumes simple matrices.   We present a novel reduction technique for a large family of problems (including eventual non-negativity/positivity, everywhere non-negativity/positivity etc.) over diagonalizable matrices to the corresponding problem over simple matrices.  This yields effective decision procedures for {\EPSum} and {\ENSum} for diagonalizable matrices. Our reduction makes use of a novel perturbation analysis that also has other interesting consequences.
  \end{enumerate}
  
  As mentioned earlier, the eventual non-negativity and positivity problem for single rational matrices are well-motivated in the literature, and {\EPosId} (or {\EPSum} with $|\AAA| = 1$) is known to be in {\myPTIME}~\cite{NOUTSOS2006132}. But for {\ENId}, no decidability results are known to the best of our knowledge. From our work, we obtain two new results about {\ENId}: (i) in general {\ENId} reduces to {\UPos} and (ii) for  diagonalizable matrices, we can decide {\ENId}. What is surprising (see Section~\ref{sec:pos}) is that the latter decidability result goes via ${\ENSum}$, i.e. the multiple matrices case. Thus, reasoning about sums of powers of matrices, viz. {\ENSum}, is useful even when reasoning about powers of a single matrix, viz. \ENId.

\subsection{Potential applications of \ENSum{} and \EPSum{}}
\label{sec:motiv}

A prime motivation for defining the generalized problem statement \ENSum{} is that it is useful even when reasoning about the single matrix case \ENId{}. However and unsurprisingly, \ENSum{} and \EPSum{} are  also well-motivated independently. Indeed, for every application involving a linear dynamical system that reduces to  \ENId{}/\EPosId{}, there is a naturally defined aggregated version of the application involving multiple independent linear dynamical systems that reduces to  \ENSum{}/\EPSum{} (e.g see the swarm of robots example that follows). Beyond this, \ENSum{}/\EPSum{} arise naturally and directly when solving problems in different practical scenarios. We give a detailed description of several such applications in this section. 

\medskip 

\noindent {\bf Partially observable multi-modal systems.}
Our first example comes from the domain of cyber-physical systems in a partially observable setting. Consider a system (e.g. a robot) with $m$ modes of operation, where the $i^{th}$ mode dynamics is given by a linear transformation encoded as a $k \times k$ matrix of rationals, say $A_i$.  Thus, if the system state at (discrete) time $t$ is represented by a $k$-dimensional rational (row) vector $\myVec{u_t}$, the state at time $t+1$, when operating in mode $i$, is given by $\myVec{u_t}A_i$. Suppose the system chooses to operate in one of its various modes at time $0$, and then sticks to this mode at all subsequent time. Further, the initial choice of mode is not observable, and we are only given a probability distribution over modes for the initial choice. This is natural, for instance, if our robot (multi-modal system) knows the terrain map and can make an initial choice of which path (mode) to take, but cannot change its path once it has chosen. If $p_i$ is a rational number denoting the probability of choosing mode $i$ initially, then the expected state at time $n$ is given by $\sum_{i=1}^m p_i\cdot \myVec{u_0}A_i^n$ $= \myVec{u_0}\big(\sum_{i=1}^m p_i\cdot A_i^n\big)$. A safety question in this context is whether starting from a state $\myVec{u_0}$ with all non-negative (resp. positive) components, the system is expected to eventually stay locked in states that have all non-negative (resp. positive) components. In other words, does $\myVec{u_0}\big(\sum_{i=1}^m p_i\cdot A_i^n\big)$ have all non-negative (resp. positive) entries for all sufficiently large $n$? Clearly, a sufficient condition for an affirmative answer to this question is to have $\sum_{i=1}^n p_i\cdot A_i^n$ eventually non-negative (resp. positive), which is an instance of {\ENSum} (resp. {\EPSum}).

\medskip

\noindent{\bf Commodity flow networks.}  Consider a flow network where $m$ different commodities $\{c_1,\ldots, c_m\}$ use the same flow infrastructure spanning $k$ nodes, but have different loss/regeneration rates along different links. For every pair of nodes $i,j\in\{1,\ldots,k\}$ and for every commodity $c\in\{c_1,\ldots,c_m\}$, suppose $A_c[i,j]$ gives the fraction of the flow of commodity $c$ starting from $i$ that reaches $j$ through the link connecting $i$ and $j$ (if it exists).  In general, $A_c[i,j]$ is the product of the fraction of the flow of commodity $c$ starting at $i$ that is sent along the link to $j$, and the loss/regeneration rate of $c$ as it flows in the link from $i$ to $j$. Note that $A_c[i,j]$ can be $0$ if commodity $c$ is never sent directly from $i$ to $j$, or the commodity is lost or destroyed in flowing along the link from $i$ to $j$.  It can be shown that $A_c^n[i,j]$ gives the fraction of the flow of $c$ starting from $i$ that reaches $j$ after $n$ hops through the network. If commodities keep circulating through the network ad-infinitum, we wish to find if the network gets \emph{saturated}, i.e., for all sufficiently long enough hops through the network, there is a non-zero fraction of some commodity that flows from $i$ to $j$ for every pair $i,j$.  This is equivalent to asking if there exists $N\in\mathbb{N}$ such that  $\sum_{\ell=1}^m A_{c_\ell}^n>0$.  If different commodities have different weights (or costs) associated, with commodity $c_i$ having the weight $w_i$, the above formulation asks if $\sum_{\ell=1}^m w_\ell.A_{c_\ell}^n$ is eventually positive, which is effectively the \EPSum{} problem. 

\noindent{\bf Eventual Occupancy for a Swarm of Robots.} The first is again from the domain i.e cyber physical systems where we consider a swarm of $m$ robots denoted by $R_i : \{1,2, \ldots, m\}$ collectively doing a job. Suppose the robots have similar but not identical dynamics (a practical reality), described by a set of  $n \times n$ rational matrices $A_i$, $i \in \{ 1, 2, \ldots, m \}$, one for each robot $R_i$. Now, we assume that each robot $R_i$ starts from the same positive initial state vector $\mathbf{u_0}$.  However, due to variation among the dynamics ($A_i$'s) of the different robots, the exact states of the various robots may not coincide as they evolve over time.  A (possibly weighted) average of the individual robot's states gives a single aggregated state of the swarm.  Properties of this aggregated measure are therefore of interest when studying the aggregate swarm dynamics. A common problem of interest in this setting is to ask if the aggregate swarm eventually occupy a desired region of space. Assuming that the swarm state captures the coordinate position, this question amounts to asking if $\mathbf{u_0}.(w_1.A_1^n + ... w_m.A_m^n)$ is eventually non-negative or positive, where all $w_i > 0$ and $w_1 + ... w_m = 1$.  A sufficient condition for this to happen is when $(w_1. A_1^n + ... w_m.A_m^n)$ is eventually non-negative or positive. The problem therefore reduces to an instance of the \ENSum{} problem where the set $\AAA = \{(1, A_1), \ldots (1, A_m)\}$ and each $A_l$ is a $n\times n$ matrix, governing the dynamics of the $i^{th}$ robot $R_i$.
\medskip 

\noindent {\bf Quantitative properties for weighted automata.}
 Our next  example comes from the study of quantitative properties of special languages defined by weighted automata~\cite{weighted-aut-book}.  Consider a weighted automaton $\mathcal{W}$ with states $\{s_1, \ldots s_k\}$ over the alphabet $\{a_1, \ldots a_m\}$. For every transition from $s_i$ to $s_j$ on letter $a_l$, suppose $w_{i,j,l}$ is the rational-valued weight associated with the transition.  The automaton $\mathcal{W}$ naturally defines the languages $P_{i,j,t}$ for all $i, j \in \{1, \ldots k\}$ and $t > 0$, where $P_{i,j,t}$ consists of all length $t$ "pure" (i.e. no mixing of letters) words that label a path from $s_i$ to $s_j$ in $\mathcal{W}$.  In other words, $P_{i,j,t}  = \{\sigma : |\sigma| = t, \sigma \in (a_1^* + \ldots a_m^*)$ and $\exists$ path from $s_i$ to $s_j$ labeled $\sigma\}$. The weight of $P_{i,j,t}$ is simply the sum of weights of all paths from $s_i$ to $s_j$ labeled by a word in $P_{i,j,t}$, where the weight of a path is the product of weights of edges along the path. An interesting question to ask in this setting is whether all $P_{i,j,t}$s ultimately have non-negative (resp. positive) weights.  It can be shown that this problem reduces to {\ENSum} (resp. {\EPSum}), where the set $\AAA = \{(1, A_1), \ldots (1, A_m)\}$ and each $A_l$ is a $k\times k$ matrix, with $A_l[i,j] = w_{i,j,l}$.

\medskip 

\noindent {\bf Continuous linear dynamical systems.} So far, we have looked at applications where the eventual properties of \emph{discrete time linear systems} were of interest. Given a rational state transformation matrix and initial state vector we were interested in the eventual non-negativity(resp. positivity) of the state vector. It would appear that the eventual non negativity problem is relevant to discrete time linear systems only. But as we discuss below, the interest in the problem of eventual non-negativity of states spans over continuous time linear dynamical systems, wherein the state transformation matrices and initial states could possibly have real entries. Some interesting applications arising out of non-negative real matrices are pointed out in \cite{berman1989nonnegative} and \cite{farina2000positive}. These applications arise out a large variety of domains : engineering, medicine, economics to name a few. While these applications argue about systems having a single state transformation matrix, large systems where multiple such real continuous time linear sub-systems interact with each other to achieve a single objective are also commonplace. It would be an interesting question to ask whether we can argue about the eventual non-negativity of the state vector of such a system. 

Any continuous time linear dynamical system takes the following form : $\mathbf{\dot{u}}(t) = A\mathbf{u}(t)$, $A \in \mathbb{R}^{k \times k}$, $\mathbf{u}(0) = \mathbf{u_0} \in \mathbb{R}^k$, $t > 0$, the solution of which can be represented by the set $\{ \mathbf{u}(t) = e^{tA}\mathbf{u_0} | t \in [0, \infty) \}$. Now consider a collection of such real continuous time dynamical systems : $\{ S_i : \{ \mathbf{\dot{u}_i}(t) = A_i\mathbf{u_i} \}, i \in [1,m] \}$ interacting with each other, starting from the same initial state $\mathbf{u_0}$ such that the aggregated dynamics of the system is given by the weighted sum of each sub-system. Moreover, the contribution of each sub-system $S_i$, is decided by a weight $w_i$ associated with it. The state vector of such a system is therefore given by $\mathbf{\dot{u}}(t) = \sum_{i=1}^{m} w_i\mathbf{\dot{u}}_i(t)$ $ = \sum_{i=1}^{m} w_iA_i\mathbf{u_i}(t)$. Now, we can write the state vector in the exponential solution form at the $t^{th}$ time instant as follows : $\mathbf{u}(t) = \sum_{i=1}^{m} w_ie^{tA_i}\mathbf{u_0}$ $=\mathbf{u_0}(\sum_{j=0}^{\infty}\frac{t^i}{i!})(w_1A_1^j + \ldots + w_mA_m^j)$. Observe that $(\sum_{j=0}^{\infty}\frac{t^i}{i!}) > 0$, therefore if the system starts at a non-negative (resp. positive) initial state vector, the eventual non-negativity of the state vector $\mathbf{u}(t)$ of the system at the $t^{th}$ time instant, is equivalent to asking whether $\exists N$ s.t $\forall n > N$, $w_1A_1^n + \ldots + w_mA_m^n \geq 0$ (resp. $> 0$). This amounts to asking whether weighted sum of the set $\AAA = \{ (w_1, A_1), \ldots, (w_m, A_m) \}$ is eventually non-negative (resp. positive), exactly an instance of the \ENSum{} (resp. \EPSum{}) problem.

\subsection{Other related work}

Our problems of interest are different from other well-studied  problems that arise if the system is allowed to choose its mode independently at each time step (e.g. as happens in Markov decision processes~\cite{KVAK10,AGV18}).  The crucial difference stems from the fact that we require that the mode be chosen  once initially, and subsequently, the system must follow the same mode forever. Thus, our problems are prima facie different from those related to general probabilistic or weighted finite automata, where reachability of states and  questions pertaining to long-run behaviour are either known to be undecidable or have remained open for long (\cite{article1}, \cite{gimbert2010probabilistic}, \cite{almagor2020s}). Even in the case of unary probabilistic/weighted finite automata ~\cite{AAOW15,bell2021decision,BFLM20,AGKV16}, reachability is known in general to be as hard as the Skolem problem on linear recurrences -- a long-standing open problem, with decidability only known in very restricted cases. The difference sometimes manifests itself in the simplicity/hardness of solutions.  For example,   {\EPosId} (or {\EPSum} with $|\AAA| = 1$) is known to be in {\myPTIME}~\cite{NOUTSOS2006132} (not so for {\ENId} however), whereas it is still open whether the reachability problem for unary probabilistic/weighted automata is decidable. 
It is also worth remarking that instead of the sum of powers of  matrices, if we considered the product of their powers, we would effectively be solving  problems akin to the \emph{mortality problem} \cite{bell2012mortality}, \cite{bell2021mortality} (which asks whether the all-$0$ matrix can be reached by multiplying with repetition from a set of matrices) -- a notoriously difficult problem. The diagonalizable matrix restriction is a common feature in in the context of linear loop programs (see, e.g, ~\cite{OuakninePWSODA15,AlmagorKKOPOPL021}), where matrices are used for updates. Finally, logics to reason about temporal properties of linear loops have been studied, although decidability is known only in  restrictive settings, e.g. when each predicate defines a semi-algebraic set contained in some $3$-dimensional subspace, or has intrinsic dimension $1$~\cite{KarimovLOPVWWPOPL22}.

\section{Preliminaries}
\label{sec:prelim}

The symbols $\mathbb{Q}, \mathbb{R}$, $\mathbb{A}$ and $\mathbb{C}$ denote the set of rational, real, algebraic and complex numbers respectively.   Recall that an \emph{algebraic number} is a root of a non-zero polynomial in one variable with rational coefficients. An algebraic number can be real or complex. We use $\mathbb{RA}$ to denote the set of real algebraic numbers (which includes all rationals). The sum, difference and product of two (real) algebraic numbers is again (real) algebraic.
Furthermore, every root of a polynomial equation with (real) algebraic coefficients is again (real) algebraic. We call matrices with all rational (resp. real algebraic or real) entries \emph{rational} (resp. \emph{real algebraic} or \emph{real}) \emph{matrices}.  We use $A \in \mathbb{Q}^{k \times l}$ (resp. $A \in \mathbb{R}^{k\times l}$ and $A \in \mathbb{RA}^{k\times l}$) to denote that $A$ is
a $k\times l$ rational (resp. real and real algebraic) matrix, with rows indexed $1$ through $k$, and columns indexed $1$ through $l$.  The entry in the $i^{th}$ row and $j^{th}$ column of a matrix $A$ is denoted $A[i,j]$.  If $A$ is a column vector (i.e. $l = 1$), we often use boldface letters, viz. $\mathbf{A}$, to refer to it. In such cases, we use $\mathbf{A}[i]$ to denote the $i^{th}$ component of $\mathbf{A}$, i.e. $A[i,1]$.  The transpose of a $k \times l$ matrix $A$, denoted $\transp{A}$, is the $l \times k$ matrix obtained by letting $\transp{A}[i,j] = A[j,i]$ for all $i \in \{1, \ldots l\}$ and $j \in \{1, \ldots k\}$.  Matrix $A$ is said to be \emph{non-negative} (resp. \emph{positive}) if all entries of $A$ are non-negative (resp. positive) real numbers.  Given a set $\AAA
= \{(w_1, A_1), \ldots (w_m, A_m)\}$ of (weight, matrix) pairs, where each $A_i \in \mathbb{Q}^{k\times k}$ (resp. $\in \mathbb{RA}^{k\times k}$) and each $w_i \in \mathbb{Q}$,
we use $\msum{\AAA}{n}$ to denote the weighted matrix sum $\sum_{i=1}^m w_i\cdot A_i^n$,
for every natural number $n > 0$.  Note that $\sum \AAA^n$ is itself a
matrix in $\mathbb{Q}^{k\times k}$ (resp. $\mathbb{RA}^{k\times k}$).  

\begin{definition}
  We say that $\AAA$ is eventually non-negative (resp. positive) iff
  there is a positive integer $N$ s.t., $\msum{\AAA}{n}$ is
  non-negative (resp. positive) for all $n \geq N$.
\end{definition}

The {\ENSum} (resp. {\EPSum}) problem, described in
Section~\ref{sec:intro}, can now be re-phrased as: \emph{Given a
set $\AAA$ of pairs of rational weights and rational $k\times k$ matrices, is $\AAA$ eventually non-negative
(resp. positive)?} As mentioned in Section~\ref{sec:intro}, if
$\AAA = \{(1, A)\}$, the {\ENSum} (resp. {\EPSum}) problem is also called {\ENId} (resp. {\EPosId}). 

We note that the study of {\ENSum} and {\EPSum} with $|\AAA| = 1$ is effectively the study of {\ENId} and {\EPosId} i.e., wlog we can assume $w=1$. Clearly,  if $w > 0$, then $w\cdot A^n$ is eventually non-negative (resp. positive) iff $A^n$ is eventually non-negative (resp. positive).  Therefore, $w$ may be taken to be $1$ in this case.  The case when $w=0$ is trivial.  Finally, if $w < 0$, $w\cdot A^n$ is eventually non-negative (resp. positive) iff there exists $N > 0$ such that all entries of $A^n$ are non-positive (resp. negative) for all $n \ge N$.  However, no rational matrix $A$ can have all entries of $A^n$ non-positive, with at least one negative entry, for all $n \ge N$.\footnote{$A^{2N} = A^{N}\times A^{N}$ must have all non-negative entries if $A^N$ has all non-positive entries. Hence, $A^n$ cannot have all non-positive entries with at least one negative entry, for all $n \ge N$.} Hence, if $w < 0$, the only interesting case amounts to asking if $A^n$ equals the all-$0$ matrix for some $n > 0$. For a $k\times k$ matrix $A$, this is equivalent to asking if for each $i \in \{1, \ldots k\}$, there exists $n_i > 0$ such that $A^{n_i} {\mathbf e_i}^T = {\mathbf 0}^T$, where $\mathbf{e_i}$ is the $k$-dimensional unit vector with its $i^{th}$ component set to $1$ and all other components set to $0$, and $\mathbf 0$ is the $k$-dimensional vector of all $0$s.  Each of these problems is an instance of the \emph{Orbit Problem}~\cite{orbit}, which is known to be decidable in {\myPTIME}.  Therefore, the only remaining interesting case  when $\AAA = \{(w, A)\}$ is if $w = 1$. This explains why the study of {\ENSum} and {\EPSum} with $|\AAA| = 1$ is effectively the study of {\ENId} and {\EPosId} respectively.

The \emph{characteristic polynomial} of a matrix
$A \in \mathbb{RA}^{k \times k}$ is given by $\mathit{det}(A-\lambda I)$, where
$I$ denotes the $k\times k$ identity matrix.  Note that this is a
degree $k$ polynomial in $\lambda$.  The roots of the characteristic
polynomial are called the \emph{eigenvalues} of $A$.  The non-zero
vector solution of the equation $A\mathbf{x} = \lambda_i \mathbf{x}$,
where $\lambda_i$ is an eigenvalue of $A$, is called
an \emph{eigenvector} of $A$.  Although  $A \in \mathbb{RA}^{k\times k}$, in general it can have eigenvalues $\lambda \in \mathbb{C}$ which are all algebraic numbers. An eigenvector is said to be positive
(resp. non-negative) if each component of the eigenvector is  a positive
(resp. non-negative) rational number.  A matrix is called \emph{simple} if all its eigenvalues are distinct. Further, a matrix $A$ is called \emph{diagonalizable} if there exists an invertible matrix $S$ and diagonal matrix $D$ such that $S^{-1} D S= A$. 

The study of weighted sum of powers of matrices is intimately related to the study of \emph{linear recurrence sequences
(LRS)}, as we shall see.
We now present some definitions and useful properties of LRS.  For more details on LRS, the reader is
referred to the work of Everest et al. \cite{d2ad832a2cc24904b7936daaa5c1c17d}.

A sequence of rational numbers $\langle u \rangle$ = $\langle u_n \rangle_{n=0}^{\infty}$ is called an LRS of \emph{order} $k ~(> 0)$ if
     the $n^{th}$ term of the sequence, for all $n \geq k$, can be
     expressed using the recurrence:
         $u_n = a_{k-1}u_{n-1} + \ldots + a_1u_{n-k-1} + a_0u_{n-k}.$ 
         
Here, $a_0 ~(\neq 0), a_1, \ldots, a_{k-1} \in \mathbb{Q}$ are called the
\emph{coefficients} of the LRS, and $u_0, u_1, \ldots, u_{k-1} \in \mathbb{Q}$ are called the
\emph{initial values} of the LRS. Given the
coefficients and initial values, an LRS is uniquely defined.  However, the same LRS may be defined
by multiple sets of coefficients and corresponding initial values.  An
LRS $\langle u \rangle$ is said to be \emph{periodic} with period $\rho$
if it can be defined by the recurrence $u_n = u_{n-\rho}$ for all $n \ge
\rho$.  
Given an LRS $\langle u \rangle$, its \textit{characteristic
  polynomial} is $p_{\langle u \rangle}(x) = x^k -
  \sum_{i=0}^{k-1}a_{i}x^{i}$. We can factorize the characteristic
  polynomial as $p_{{\langle u \rangle}}(x) =
  \prod_{j=1}^d(x-\lambda_j)^{\rho_j}$, where $\lambda_j$ is a root,
  called a \textit{characteristic root} of \emph{algebraic multiplicity} $\rho_j$. An LRS
  is called \emph{simple} if $\rho_j = 1$ for all $j$, i.e. all
  characteristic roots are distinct. Let $\{\lambda_1,
  \lambda_2, \ldots , \lambda_d \}$ be distinct roots of
  $p_{\langle u \rangle}(x)$ with multiplicities $\rho_1, \rho_2, \ldots ,
  \rho_d$ respectively. Then the $n^{th}$ term of the LRS,
  denoted $u_n$, can be expressed as $u_n =
  \sum_{j=1}^{d}q_j(n)\lambda_j^n$, where $q_j(x) \in \mathbb{C}(x)$
  are univariate polynomials of degree at
  most $\rho_j - 1$  with complex coefficients such that $\sum_{j=1}^d\rho_j = k$. This representation of an LRS is known as the \textit{exponential polynomial solution} representation.
  
 \begin{wrapfigure}[5]{l}{0.35\textwidth} 
  \begin{minipage}{0.3\textwidth}
    \vspace*{-1.85cm}
  \begin{equation*}
    M_{\langle u \rangle} = \begin{bmatrix}
        a_{k-1} & 1 & \ldots & 0 & 0 \\
        \vdots & \vdots & \ddots & \vdots & \vdots \\
        a_2 & 0 & \ldots & 1 & 0 \\
        a_1 & 0 & \ldots & 0 & 1 \\
        a_0 & 0 & \ldots & 0 & 0
    \end{bmatrix}
  \end{equation*}
\vspace*{-0.7cm}
\end{minipage}
\caption{Companion matrix}
\label{fig:comp-matrix}
\end{wrapfigure}
\vspace*{0.8cm} It is well known that scaling an LRS by a constant gives another LRS, and
  the sum and product of two LRSs is also an LRS (Theorem 4.1
  in~\cite{d2ad832a2cc24904b7936daaa5c1c17d}). Given an LRS $\langle u \rangle$ defined by $ u_n = a_{k-1}u_{n-1} +
\ldots + a_1u_{n-k-1} + a_0u_{n-k}$, we define its \emph{companion
  matrix} $M_{\langle u \rangle}$ to be the $k \times k$ matrix shown in Fig.~\ref{fig:comp-matrix}.
When $\langle u \rangle$ is clear from
the context, we often omit the subscript for clarity of notation, and
use $M$ for $M_{\langle u \rangle}$. 
Let $\mathbf{u} = (u_{k-1},
\ldots, u_0)$  be a row vector containing the
$k$ initial values of the recurrence, and let $\mathbf{e_k} = (0,
0, \ldots 1)^T$ be a column vector of $k$ dimensions with the last
element equal to $1$ and the rest set to $0$s. It is easy to see that for all $n \ge 1$, $\mathbf{u} M^n \mathbf{e_k}$ gives $u_n$.  Note that the eigenvalues of the matrix $M$ are exactly the roots of the characteristic polynomial of the
LRS $\langle u \rangle$.

For $\mathbf{u} = (u_{k-1}, \ldots, u_0)$, we call the matrix
$G_{\langle u \rangle}= \begin{bmatrix} 0 & \mathbf{u} \\
             \mathbf{0}^T & M_{\langle u \rangle}
\end{bmatrix}$
the \emph{generator matrix} of the LRS $\langle u \rangle$, where
$\mathbf{0}$ is a $k$-dimensional vector of all $0$s. We omit the subscript and use $G$ instead of $G_{\langle u
  \rangle}$, when the LRS $\langle u \rangle$ is clear from the
context.    It is easy to show from the
above that $u_n = G^{n+1}[1,k+1]$ for all $n \geq 0$.

We say that an LRS $\langle u \rangle$ is
\emph{ultimately non-negative} (resp. \emph{ultimately positive})
iff there exists $N > 0$, such that $\forall n \ge N$, $u_n \geq 0$
(resp. $u_n > 0$)\footnote{\emph{Ultimately non-negative} (resp. \emph{ultimately positive})
LRS, as defined by us, have also been called \emph{ultimately positive}
(resp. \emph{strictly positive}) LRS elsewhere in the
literature \cite{ouaknine2017ultimate}.  However, we choose to use terminology that is consistent across matrices and LRS, to avoid notational
confusion.}.

The problem of determining whether a given LRS is ultimately
non-negative (resp. ultimately positive) is called the \emph{Ultimate
Non-negativity} (resp. \emph{Ultimate Positivity}) problem for LRS.
We use {\UPos}
(resp. {\SUPos}) to refer to this problem.  It is
known~\cite{halava2005skolem} that {\UPos} and {\SUPos} are
polynomially inter-reducible, and these problems have been widely studied in the literature (e.g., \cite{DBLP:conf/fct/Ouaknine13}, \cite{ouaknine2017ultimate}, \cite{DBLP:journals/siglog/OuaknineW15}). 
A closely related problem is the \emph{Skolem problem}, wherein
we are given an LRS $\langle u \rangle$ and we are required to
determine if there exists $n \geq 0$ such that $u_n = 0$.  The relation
between the Skolem problem and {\UPos} (resp. {\SUPos}) has been
extensively studied in the literature (e.g.,~\cite{halava2005skolem}, ~\cite{halava2006positivity}, ~\cite{10.1145/301250.301389}).

\section{Hardness of eventual non-negativity \& positivity}
\label{sec:ENIdhard}

 In this section, we show that  \UPos{} (resp. \SUPos{})  polynomially reduces to  \ENSum{} (resp. \EPSum{}) when $|\AAA| \ge 2$.  Since  \UPos{} and \SUPos{} are known to be $\mathsf{coNP}$-hard (in fact, as hard as the decision problem for the universal theory of reals Theorem 5.3~\cite{ouaknine2017ultimate}), we conclude that {\ENSum} and {\EPSum} are also {\mycoNP}-hard and at least as hard as the decision problem for the universal theory of reals, when $|\AAA| \ge 2$.  Thus, unless $\mathsf{P=NP}$, there is no hope of finding polynomial-time solutions to these problems. 

\begin{theorem}
\label{lma : UNeg_red_ENSum}
    {\UPos} reduces to {\ENSum} with $|\AAA| \ge 2$ in polynomial time. 
\end{theorem}

\begin{proof}
    Given an LRS $\langle u \rangle$ of order $k$ defined by the recurrence $u_n = a_{k-1}u_{n-1} + \ldots + a_1u_{n-k-1} + a_0u_{n-k}$, we show how to construct two   $k\times k$ matrices $A_1$ and $A_2$ such that $\langle u \rangle$ is ultimately non-negative iff $(A_1^n + A_2^n)$ is  eventually non-negative i.e we reduce \UPos{} for $\langle u \rangle$ to \ENSum{} for $\AAA = \{ (1,A_1), (1, A_2) \}$ .  
    
    Consider $A_1$ to be the generator matrix of $\langle u \rangle$, i.e. $A_1= \begin{bmatrix} 0 & \mathbf{u} \\ \mathbf{0}^T & M \end{bmatrix}$ where $M$ is the companion matrix of $\langle u \rangle$ (see Fig.~\ref{fig:comp-matrix}) and $\mathbf{u} = (u_{k-1}, \ldots, u_1, u_0)$ is the vector of initial values of $\langle u \rangle$. With abuse of notation, let $|x|$ denote the absolute value of a rational number $x$. Now define a number $a_{max} = max(|a_0|,|a_1|, \ldots, |a_{k-1}|)$, where $a_0, a_1,\ldots,a_{k-1}$ are the coefficients of the recurrence defining the LRS $\langle u \rangle$. Construct a matrix $P$ from matrix $M$ by replacing each coefficient of $\langle u \rangle$ in the first column of $M$ (see Fig.~\ref{fig:comp-matrix}) by $a_{max}$,  and by keeping the remaining entries of $M$ unchanged. Clearly, $P$ can be constructed in time polynomial in the size of the given LRS instance $\langle u \rangle$. Finally, construct the matrix $A_2$ as follows: $A_2= \begin{bmatrix} 0 & \mathbf{0} \\ \mathbf{0}^T & P \end{bmatrix}$. 
    Now consider the sequence of matrices defined by $A_1^n + A_2^n$, for all $n \geq 1$. By properties of the generator matrix, it is easily verified that $ A_1^n = \begin{bmatrix} 0 & \mathbf{u}M^{n-1} \\ \mathbf{0}^T & M^n \end{bmatrix}$.  In a similar manner, it can be verified that $A_2^n = \begin{bmatrix} 0 & \mathbf{0} \\ \mathbf{0}^T & P^n \end{bmatrix}$.  Therefore,
$A_1^n + A_2^n =$ $\begin{bmatrix} 0 & \mathbf{u}M^{n-1} \\ \mathbf{0}^T & P^n + M^n \end{bmatrix}$, for all $n \ge 1$.
    
    We claim  that the $n \times n$ matrix $P^n + M^n$ is always non-negative. To see why this is so, notice that $P$ has only non-negative entries, and  each entry is at least as large as the absolute value of the corresponding entry of $M$ i.e. $P[i,j] \geq |M[i,j]| \ge 0$ for all $i, j \in \{1, \ldots k\}$. Hence, $P^n[i,j] \geq |M^n[i,j]|\geq -M^n[i,j]$, and  $P^n[i,j] + M^n[i,j] \ge 0$ for all $i, j \in \{1, \ldots k\}$ and for all $n \ge 1$. Finally, we show that $A_1^n + A_2^n$ is eventually non-negative iff $\langle u \rangle$ is ultimately non-negative, which would complete the proof.
    
    Suppose $\langle u \rangle$ is ultimately non-negative. Observe that the elements of the vector $\mathbf{u}M^{n-1}$ are essentially $(u_{n+k-2} \ldots,u_{n},u_{n-1})$. Hence, if $\langle u \rangle$ is ultimately non-negative, there exists $N \geq 1$ such that the vector $\mathbf{u}M^{n-1} = (u_{n+k-2} \ldots,u_{n},u_{n-1})$ has all non-negative components, for all $n \geq N$. Thus, for all $n \geq N$, the elements of the matrix $(A_1^n + A_2^n)$ are non-negative. In other words,  $(A_1^n + A_2^n)$ is eventually non-negative.
    
    Now, suppose $(A_1^n + A_2^n)$ is eventually non-negative. Recall from above that $A_1^n + A_2^n = \begin{bmatrix} 0 & \mathbf{u}M^{n-1} \\ \mathbf{0}^T & M^n + P^n \end{bmatrix}$.  We have already seen  that all entries of $M^n + P^n$ are non-negative for all $n \geq 1$. Therefore, since $(A_1^n + A_2^n)$ is eventually non-negative, there must exist $N > 0$ such that all components of $\mathbf{u}M^{n-1}$ are non-negative for all $n \geq N$. However, we also know that $\mathbf{u}M^{n-1} = (u_{n+k-2} \ldots,u_{n},u_{n-1})$.  Hence, all elements in the sequence $\langle u \rangle$, starting from $u_{N - 1}$ are non-negative.  This implies that the LRS $\langle u \rangle$ is ultimately non-negative. \qed
  
\end{proof}

Observe that the same reduction technique works if we are required to use more than 2 matrices in {\ENSum}. Indeed, we can construct matrices $A_3, A_4, \ldots, A_m$ similar to the construction of $A_2$ in the reduction above, by having the $k \times k$ matrix in the bottom right (see definition of $A_2$) to have positive values greater than the maximum absolute value of every element in the companion matrix.

A simple modification of the above proof setting $A_2 = \begin{bmatrix} 1 & \mathbf{0} \\ \mathbf{1}^T & P \end{bmatrix}$, where $\mathbf{1}$ denotes the $k$-dimensional vector of all $1$'s gives us the corresponding hardness result for \EPSum{}. We formalize this same as a theorem below : 

\begin{restatable}{theorem}{UPredEPSum}
\label{lma : UP_red_EPSum}
    {\SUPos} reduces to {\EPSum} with $|\AAA| \ge 2$ in polynomial time. 
\end{restatable}

\begin{proof}
     Given an LRS $\langle u \rangle$ of order $k$ .We construct a set of two matrices $A_1$ and $A_2$, such that $\langle u \rangle$ is ultimately non-negative iff the sum of powers of $A_1$ and $A_2$ are eventually non-negative. Consider $A_1$ to be the generator matrix of $\langle u \rangle$ i.e $A_1= \begin{bmatrix} 0 & \mathbf{v} \\ \mathbf{0}^T & M \end{bmatrix}$ where $M$ is the companion matrix of $\langle u \rangle$ and $\mathbf{v} = (u_{k-1}, \ldots, u_1, u_0)$ is the vector containing the set of initial values of $\langle u \rangle$. Define a number $a_{max}$ as follows : $a_{max} = max(|a_0|,|a_1|, \ldots, |a_{k-1}|)$ where $a_0,a_1,\ldots,a_{k-1}$ are the coefficients of the LRS $\langle u \rangle$(Here, $|x|$ denotes the absolute value of a rational number $x$). Construct a matrix matrix $P$ by replacing the coefficients of $\langle u \rangle$ from $M$ by some element greater than $a_{max}$. $P$ can be constructed in polynomial time w.r.t $max(|a_0|,|a_1|, \ldots, |a_{k-1}|)$. Construct $P$, such that each element of $P$ is positive. Now, construct the matrix $A_2$ as follows : $A_2= \begin{bmatrix} 1 & \mathbf{0} \\ \mathbf{1}^T & P \end{bmatrix}$. Now, consider the matrix $A(n) = A_1^n + A_2^n$ for all $n \geq 1$. $A(n)= \begin{bmatrix} 1 & \mathbf{v}M^{n-1} \\ \mathbf{1}^T.P^{n-1} & P^n + M^n \end{bmatrix}$(since, $ A_1^n = \begin{bmatrix} 0 & \mathbf{v}M^{n-1} \\ \mathbf{0}^T & M^n \end{bmatrix}$ by the properties of the generator matrix). 
    
    Observe that the $n \times n$ matrix $P^n + M^n$ is always positive. Since $P$ is a matrix with only positive entries where each entry is greater than the corresponding entry of $M$ i.e $P[i,j] > M[i,j]$. Hence, $P^n[i,j] > M^n[i,j]$, and no entry of the sum $P^n[i,j] + M^n[i,j]$ can therefore become zero or negative. Using this fact now, we show that the sum of powers of the set of matrices $A_1 = \{A_1, A_2\}$ is eventually positive iff $\langle u \rangle$ is ultimately positive. 
    
    Consider $\langle u \rangle$ is ultimately positive. Observe that the elements of the vector $\mathbf{v}M^{n-1}$ are essentially $(u_{n+k-2} \ldots,u_{n},u_{n-1})$. Hence, if $\langle u \rangle$ is ultimately positive, there exists $n_0 \geq 1$ such that the vector  $\mathbf{v}M^{n-1} = (u_{n+k-2} \ldots,u_{n},u_{n-1}) > 0$ for all $n \geq n_0$. Thus, for all $n \geq n_0$, the elements of the matrix is $A(n)$ are greater than $0$. Therefore if $\langle u \rangle$ is ultimately positive, the set of matrices $\{A_1, A_2\}$ is eventually positive.
    
    Now, consider the set of matrices $\{A_1, A_2\}$ to be eventually positive. We have $ A_1^n = \begin{bmatrix} 0 & \mathbf{v}M^{n-1} \\ \mathbf{0}^T & M^n \end{bmatrix}$ and $A_2^n = \begin{bmatrix} 1 & \mathbf{0} \\ \mathbf{1}^T.P^{n-1} & P^n \end{bmatrix}$. For all $n \geq n_0$, we have $A(n) = A_1^n + A_2^n$. We have already observed that $M^n + P^n$ is positive for all $n \geq 1$. Now, since, $A(n)$ is eventually positive, we have that $\mathbf{v}M^{n-1} > \mathbf{0}$ for all $n \geq n_0$. i.e for all $n \geq n_0$, the vector $\mathbf{v}M^{n-1} = (u_{n+k-2} \ldots,u_{n},u_{n-1}) > 0$. Hence, the sequence, $u_{n_0 - 1}, u_{n_0}, \ldots, \ldots$ is a positive sequence. Therefore, for all $n \geq n_0$, the LRS $\langle u \rangle$ is positive which implies that $\langle u \rangle$ is ultimately positive. The set of matrices $\{A_1, A_2\}$ being eventually positive implies that the LRS $\langle u \rangle$ is ultimately positive.
    
    Hence, {\SUPos} reduces to {\EPSum} in polynomial time as can be seen from the construction above.
\end{proof}

We remark that for the reduction technique used in Theorems \ref{lma : UNeg_red_ENSum} and \ref{lma : UP_red_EPSum} to work, we need at least two (weight, matrix) pairs in $\AAA$. To show why this reduction doesn't work when $|\AAA|=1$, we give a counterexample. Consider the following case : 

Let's say we have the LRS $\langle u \rangle$ defined by the recurrence $u_n = 2u_{n-1} - u_{n-2}$, with initial values $u_1 = 0, u_0 = -1$. The sequence of elements in the recurrence is $-1,0,1,2,3,\ldots$. It is easy to show inductively that this is a strictly increasing sequence. Hence, the LRS is ultimately positive. What happens if we apply the  reduction technique used in  the proofs of Theorems \ref{lma : UNeg_red_ENSum} and \ref{lma : UP_red_EPSum} to the {\UPos} problem for this specific LRS, while restricting  $|\AAA|$ to be $1$? Let the generator matrix of $\langle u \rangle$ be denoted by $A$.  We ask if $A$ is eventually non-negative.  Recalling the definition of a generator matrix, it follows that the first row of $A^n$ is eventually non-negative iff the LRS $\langle u \rangle$ is ultimately non-negative. Moreover, the first column of $A^n$ is always non-negative (in fact, has all $0$ entries) for all $n > 0$. Hence, eventual non-negativity of $A$ depends only on the eventual non-negativity of the companion matrix $M \in \mathbb{Q}^{k \times k}$ that occupies the bottom right corner of $A$ (In this case, $M$ is a $2 \times 2$ matrix). Observe that for this specific example, $M = \begin{bmatrix} 2 & 1 \\ -1 & 0 \end{bmatrix}$ and $M^n =  \begin{bmatrix} n+1 & n \\ -n & 1-n \end{bmatrix} $.  Clearly, $M^n$ has at least one negative entry for all $n > 0$, and therefore, $A^n$ cannot be eventually non-negative.  This failure only shows that the technique used in the proofs of Theorems \ref{lma : UNeg_red_ENSum} and \ref{lma : UP_red_EPSum} cannot be directly used to prove hardness of {\ENId} or {\EPosId}. Indeed, \EPosId{} can be solved in polynomial time and showing hardness of \ENId{} in the general case is still an open problem.

Having shown the hardness of \ENSum{} and \EPSum{} when $|\AAA| \ge 2$, we now proceed to establish upper bounds on the computational complexity of  these problems.

\section{Upper bounds on eventual non-negativity \& positivity}
\label{sec : ENId reduces to Ultimate Positivity}

In this section, we show that \ENSum{} (resp. \EPSum{}) is polynomially reducible to \UPos{} (resp. \SUPos{}), regardless of $|\AAA|$. The central result of this section is formalized in the following theorem.
\begin{theorem}
\label{thm : ENSum to UPos}
    {\ENSum}, reduces to {\UPos} in polynomial time.
\end{theorem}

We complete the proof in two parts. First, we show that for a single matrix $A$, we can construct a linear recurrence $\langle a \rangle$ such that $A$ is eventually non-negative iff $\langle a \rangle$ is ultimately non-negative. Then, we show that  starting from such a linear recurrence for each matrix in $\AAA$, we can construct a new LRS, say $\langle a^\star \rangle$, with the property that the weighted sum of powers of the matrices in $\AAA$ is eventually non-negative iff $\langle a^\star \rangle$ is ultimately non-negative.  Our proof makes crucial use of the following  property of matrices.
\begin{lemma}[Adapted from Lemma 1.1 of~\cite{halava2005skolem}]
\label{lma:halava}
Let $A \in \mathbb{Q}^{k\times k}$ be a rational matrix with characteristic polynomial  $p_A(\lambda) = det(A - \lambda I)$.  Suppose we define the sequence $\langle a^{ij} \rangle$ for every $1 \leq i,j \leq k$ as follows: $a^{i,j}_n = A^{n+1}[i,j]$, for all $n \ge 0$.  Then $\langle a^{i,j} \rangle$ is an LRS of order $k$ with characteristic polynomial $p_A(x)$ and initial values given by $a^{ij}_0 = A^1[i,j], \ldots a^{ij}_{k-1} = A^k[i,j]$.
\end{lemma}
The proof of this Lemma follows from the Cayley-Hamilton Theorem and the reader is referred to~\cite{halava2005skolem} for further details.

From Lemma~\ref{lma:halava}, it is easy to see that the LRS $\langle a^{i,j} \rangle$ for all $ 1 \leq i, j \leq k$ share the same order and characteristic polynomial (hence the defining recurrence) and differ only in their initial values.  For notational convenience, in the following discussion, we say that the LRS $\langle a^{i,j} \rangle$ is \emph{generated by $A[i,j]$}.
   
\begin{restatable}{proposition}{lmaENredUPosksq}
    \label{clm : EN red UPos for k^2 LRS's}
    A matrix $A \in \mathbb{Q}^{k \times k}$ is eventually non-negative iff all LRS $\langle a^{i,j} \rangle$ generated by $A[i,j]$ for all $1 \leq i,j \leq k$ are ultimately non-negative. 
\end{restatable}
    
The proof follows easily from the definition of eventually non-negative matrices and from the definition of the LRS $\langle a^{ij} \rangle$ as follows : 
    
\begin{proof}
Let us begin with the following claim.
    \begin{claim}
    \label{lma : Matrix -> LRS}
        Each element $A[i,j]$ of a given matrix $M \in \mathbb{Q}^{k \times k}$ generates a Linear Recurrence Sequence(LRS) of order $k$.
    \end{claim}
    
    \begin{proof}
    Let $p(x) = x^k -b_{k-1}x^{k-1}-\ldots-b_1x-b_0$ be the characteristic polynomial of Matrix $A$. Then, according to the Cayley Hamilton Theorem, we have : $A^k = b_{k-1}A^{k-1} + \ldots +b_1A +b_0I$. Consequently multiplying both sides of the equation by $A^{n-k}$, we obtain : $ A^n = b_{k-1}A^{n-1} + \ldots + b_1A^{n-k+1} + b_0A^{n-k}, \forall n \geq k$. Denote by $\mathbf{e_i}$ a row vector consisting of all zero entries except the $i^{th}$ entry. Note that $A[i,j] = \mathbf{e_i} A \mathbf{e_j^T}$. Multiplying the above equation with the vectors $\mathbf{e_i}$ and $\mathbf{e_j^T}$ on both sides, it follows that : $ \mathbf{e_i} A^n \mathbf{e_j^T} = \mathbf{e_i} b_{k-1}A^{n-1} \mathbf{e_j^T}  + \ldots + \mathbf{e_i}b_1 A^{n-k+1} \mathbf{e_j^T} + \mathbf{e_i} b_0A^{n-k} \mathbf{e_j^T}$. Now, denote by $a^{ij}_r$, the element corresponding to $A^r[i,j]$. Observe that the sequence $\{a^{ij}_r : r > 0 \}$ is clearly a linear recurrence sequence of order $k$, with coefficients $b_0, b_1, \ldots, b_{k-1}$ and initial values $a^{ij}_1, a^{ij}_2 \ldots , a^{ij}_k$. Therefore, for a given matrix $A$ of order $k$, for all $1 \leq i,j \leq k$ we can construct a LRS(denoted by $\langle a^{ij} \rangle$) with the equation : $  a^{ij}_n = b_{k-1}a^{ij}_{n-1} + \ldots + b_1a^{ij}_{n-k+1} + b_0a^{ij}_{n-k}$.
    
    \end{proof}

    From the above claim, we know that a matrix $A \in \mathbb{Q}^{k \times k}$ generates a set of $k^2$ rational LRS's $\{ \langle a^{ij} \rangle : 1 \leq i,j \leq k \}$. Also, it follows from above that $a^{ij}_r = A^r[i,j]$. We prove the proposition as follows :
    
    Consider each LRS in the set $S = \{\langle a^{ij} \rangle : 1 \leq i,j \leq k \}$ to be ultimately non-negative. Hence for each $1 \leq i,j \leq k$, there exists a $\mu_{ij}$ such that $a^{ij}_n \geq 0$ for all $n \geq \mu_{ij}$. Denote $\nu = max_{1 \leq i,j \leq k}(\mu_{ij})$. Observe that for each $1 \leq i,j \leq k$, $a^{ij} \geq 0$ for all $n \geq \nu$. We know from the construction shown above that, $a^{ij}_n = A^n[i,j]$. Therefore, for all $1 \leq i,j \leq k$, $A^n[i,j] \geq 0$ for all $n \geq \nu$. Hence, $A$ is eventually non-negative.
    
    Now, consider $A$ to be eventually non-negative. Hence, there exists $n_0 \in \mathbb{N} - {0}$ such that $A^n \geq 0$ for all $n \geq n_0$ i.e for all $1 \leq i,j \leq k$ $A^n[i,j] \geq 0$ for all $n \geq n_0$. From the construction shown above we know $a^{ij}_n = A^n[i,j]$. Hence, for all $1 \leq i,j \leq k$, $a^{ij}_n \geq 0$ for all $n \geq n_0$. Therefore, each LRS in the set $S = \{\langle a^{ij} \rangle : 1 \leq i,j \leq k \}$ is ultimately positive, completing the proof.

\end{proof}

Next we define the notion of interleaving of LRS.
\begin{definition}
\label{def:lrs-int}
    Consider a set $S = \{ \langle u^i \rangle : 0 \leq i < t \}$ of $t$ LRSes, each having order $k$ and the same characteristic polynomial. An LRS $\langle v \rangle$ is said to be the \textbf{LRS-interleaving} of $S$ iff $v_{tn + s} = u^s_n$ for all $n \in \mathbb{N}$ and $0 \leq s < t$.
\end{definition}
    
Observe that, the order of $\langle v \rangle$ is $t k$ and its initial values are given by the interleaving of the $k$ initial values of the LRSes $\langle u^i \rangle$.  Formally, the initial values are $v_{tj+i} = u^i_j$ for $0 \leq i < t $ and $ 0 \leq j < k$.  The characteristic polynomial $p_{\langle v \rangle}(s)$ is equal to $p_{\langle u^i \rangle}(x^t)$.

\begin{restatable}{proposition}{interleaf}
\label{lma : interleaf}
    The LRS-interleaving $\langle v \rangle$ of a set of LRSes $S = \{ \langle u^{i} \rangle : 0 \leq i < t \}$ is ultimately non-negative iff each LRS $\langle u^i \rangle$ in $S$ is ultimately non-negative.
\end{restatable}
    
\begin{proof}
    Consider $\langle v \rangle$ is ultimately non-negative. Thus, there exists $r_0 \in \mathbb{N}$ such that for all $r \geq r_0$ $v_r \geq 0$. For all $r \in \mathbb{N}$ We can write $r = t \cdot n + s$ where $n \in \mathbb{N}$ and $0 \leq s < t$ by Euclid's division lemma. Similarly we write $r_0 = t \cdot n_0 + s_0$ where $0 \leq s_0 < t$ and $n_0 \in \mathbb{N}$. Observe that $v_{r_0} = v_{t \cdot n_0 + s_0} = u^{s_0}_{n_0}$. Since, for all $r \geq r_0$ $v_r \geq 0$, from the definition of $\langle v \rangle$, we observe that $u^i_n \geq 0$ for all $n \geq n_0$ where $s_0 \leq i < t$ and $u^j_n \geq 0$ for all $n \geq n_0 + 1$ where $0 \leq j < s_0$. Hence, each LRS in the set $S = \{ \langle u^i \rangle : 0 \leq i < t \}$ is ultimately non-negative.
    
    Now consider the set $S = \{ \langle u^i \rangle : 0 \leq i < t \}$ of LRS's to be ultimately non negative. Thus, for all $0 \leq i < t$, there exists an $n_i > 0$ such that $u^i_n \geq 0$ for all $n \geq n_i$. Following the definition of $\langle v \rangle$, we have : $\bigwedge_{i=0}^{t-1} (v_n \geq 0$ $\forall n \geq t \cdot n_i + s_i, 0 \leq s_i < t)$. Further, we can observe that in fact $v_n \geq 0$ for all $n > N$ where $N = max_{0 \leq i < t} (t \cdot n_i + s_i)$. Hence, $\langle v \rangle$ is ultimately non-negative.
\end{proof}
    
Now, from the definitions of LRSes $\langle a^{i,j} \rangle$, $\langle u^i \rangle$ and $\langle v \rangle$, and from Propositions~\ref{clm : EN red UPos for k^2 LRS's} and \ref{lma : interleaf}, we obtain the following crucial lemma.
    
\begin{lemma}
\label{cor : ENId for M -> LRS of order k^3}
    Given a matrix $A \in \mathbb{Q}^{k \times k}$, let $S = \{\langle u^{i} \rangle \mid u_n^{i} = a_n^{pq},$ where $p = \lfloor i/k \rfloor +1, ~q = i\mod k + 1, ~0 \le i < k^2\}$ be the set of $k^2$ LRSes mentioned in Lemma~\ref{lma:halava}. The LRS $\langle v \rangle$ generated by LRS-interleaving of $S$ satisfies the following: 
    \begin{enumerate}
        \item $A$ is eventually non-negative iff $\langle v \rangle$ is ultimately non-negative.
        
        \item $p_{\langle v \rangle}(x) = \prod_{i=1}^k (x^{k^2} - \lambda_i)$, where $\lambda_1, \ldots \lambda_k$ are the (possibly repeated) eigenvalues of $A$.
    
        \item $v_{rk^2 + sk + t} = u_r^{sk+t} = a_r^{s+1,t+1} = A^{r+1}[s+1,t+1]$ for all $r \in \mathbb{N}$,  $0 \leq s,t < k$.
    \end{enumerate}
\end{lemma}

We lift this argument from a single matrix to a weighted sum of matrices. 

\begin{restatable}{lemma}{lmaENSumtoUPos}
    \label{lma : ENSum_g2_to_UPos}
    Given $\AAA = \{(w_1, A_1), \ldots, (w_m, A_m)\}$, there exists a linear recurrence $\langle a^\star \rangle$, such that $\sum_{i=1}^m w_i A_i^n$ is eventually non-negative iff $\langle a^\star \rangle$ is ultimately non-negative.
\end{restatable}

\begin{proof}
    For each matrix $A_i$ in $\AAA$, let $\langle v^i \rangle$ be the interleaved LRS as constructed in Lemma~\ref{cor : ENId for M -> LRS of order k^3}.  Let $w_i \langle v^i \rangle$ denote the scaled LRS whose $n^{th}$ entry is $w_i v^i_n$ for all $n \ge 0$.  The LRS $\langle a^\star \rangle$ is obtained by adding the scaled LRSes $w_1\langle v^1\rangle, w_2\langle v^2 \rangle, \ldots$ $w_m\langle v^m \rangle$.  Clearly, $a^\star_n$ is non-negative iff $\sum_{i=1}^m w_i v^i_n$ is non-negative.  From the definition of $v^i$ (see Lemma~\ref{cor : ENId for M -> LRS of order k^3}), we also know that for all $n \ge 0$, $v^i_n = A_i^{r+1}[s+1,t+1]$, where $r = \lfloor n/k^2 \rfloor$, $s = \lfloor (n \mod k^2)/k \rfloor$ and $t = n \mod k$.  Therefore,  $a^\star_n$ is non-negative iff $\sum_{i=1}^m w_i A_i^{r+1}[s+1,t+1]$ is non-negative.  It follows that $\langle a^\star \rangle$ is ultimately non-negative iff $\sum_{i=1}^m w_i A_i^n$ is eventually non-negative. 
\end{proof}

From Lemma \ref{lma : ENSum_g2_to_UPos}, we can conclude the main result of this section, i.e., proof of Theorem \ref{thm : ENSum to UPos}. The following corollary can be shown \emph{mutatis mutandis}.

\begin{corollary}
\label{cor : EPSum -> UPos}
    {\EPSum} reduces to {\SUPos} in polynomial time.
\end{corollary}

\paragraph{Remark : } It is also possible to argue about a the eventual non-negativity (positivity) of only certain indices of the matrix using a similar argument as above. By interleaving only the LRS's corresponding to certain indices of the matrices in $\AAA$, we can show that this problem is equivalent to \UPos{} (\SUPos{}).

\section{Decision Procedures for Special Cases}
\label{sec:pos}

Since there are no known algorithms for solving {\UPos} in general, the results of the previous section present a bleak picture for deciding {\ENSum} and {\EPSum}. We now show that these problems can be solved in some important special cases. 

\subsection{Simple matrices and matrices with real algebraic  eigenvalues}
Our first positive result follows from known results for special classes of LRSes.  

\begin{restatable}{theorem}{consqENSumtoUPos}
\label{thm:consqENSumtoUPos}
    \ENSum{} and \EPSum{} are decidable for $\AAA = \{(w_1, A_1), \ldots (w_m, A_m)\}$ if one of the following conditions holds for all $i \in \{1, \ldots m\}$.
    \begin{enumerate}
    \item All $A_i$ are simple. In this case, \ENSum{} and \EPSum{} are in $\mathsf{PSPACE}$. Additionally, if the rank $k$ of all $A_i$ is fixed, \ENSum{} and \EPSum{} are in $\mathsf{PTIME}$.
    \item All eigenvalues of $A_i$ are roots of real algebraic numbers. In this case, \ENSum{} and \EPSum{} are in $\mathsf{coNP^{PosSLP}}$ (a complexity class in the Counting Hierarchy, contained in $\mathsf{PSPACE}$).
    \end{enumerate}
\end{restatable}

\begin{proof}
     Suppose each $A_i \in \mathbb{Q}^{k\times k}$, and let $\lambda_{i,1}, \ldots \lambda_{i,k}$ be the (possibly repeated) eigenvalues of $A_i$.
     The characteristic polynomial of $A_i$ is $p_{A_i}(x) = \prod_{j=1}^k (x-\lambda_{i,j})$.  Denote the LRS obtained from $A_i$ by LRS interleaving as in Lemma~\ref{cor : ENId for M -> LRS of order k^3} as $\langle a^{i} \rangle$.  By Lemma \ref{cor : ENId for M -> LRS of order k^3}, we have (i) $a^{i}_{rk^2 + sk + t} = A_i^{r+1}[s+1,t+1]$ for all $r \in \mathbb{N}$ and $0 \le s, t < k$, and (ii) $p_{\langle a^{i} \rangle}(x) = \prod_{j=1}^k \big(x^{k^2}-\lambda_{i,j}\big)$.  We now define the scaled LRS $\{\langle b^{i} \rangle$, where $\mid$ $b^{i}_n = w_i~a^{i}_n$ for all $n \in \mathbb{N}$\}.  Since scaling does not change the characteristic polynomial of an LRS (see Proposition~\ref{lma:same_char} below for a simple proof), we have $p_{\langle b^{i} \rangle}(x) = \prod_{j=1}^k \big(x^{k^2}-\lambda_{i,j}\big)$.
     
     \begin{restatable}{proposition}{lmaSameChar}
     \label{lma:same_char}
          For every LRS $\langle u \rangle$, with characteristic polynomial $p(x)$. The characteristic polynomial of an arbitrary LRS $\langle v \rangle$ defined as $v_i = wu_i$ for any arbitrary real constant $w$ is also $p(x)$.  
     \end{restatable}

    \begin{proof}
        Let the recurrence relation of $\langle u \rangle$ is given by $u_n = \sum_{i=1}^{k} a_{k-i}u_{n-i}$. We define $\langle v \rangle$ such that $v_n = wu_n$. Therefore, each term of $\langle v \rangle$ can also be written as $v_n = w\sum_{i=1}^{k} a_{k-i}u_{n-i}$ $= \sum_{i=1}^{k} a_{k-i}w u_{n-i}$ $= \sum_{i=1}^{k} a_{k-i}v_{n-i}$ (since $v_i = wu_i$ for all $i>0$). Observe that the recurrence relations of $\langle u \rangle$ and $\langle v \rangle$ are exactly the same, and hence they must have the same characteristic polynomial. 
    \end{proof}
     
     Once the LRSes $\langle b^1 \rangle, \ldots \langle b^m \rangle$ are obtained as above, we sum them to obtain the LRS $\langle b^\star \rangle$. Thus, for all $n \in \mathbb{N}$, we have $b^\star_n = \sum_{i=1}^m b^{i}_n = \sum_{i=1}^m w_i~a^{i}_n$ $= \sum_{i=1}^m w_i~A_i^r[s,t]$, where $n = rk^2 + sk + t$, $r \in \mathbb{N}$ and  $0 \le s, t < k$.  Hence, \ENSum{} (resp. \EPSum{}) for $\{(w_1, A_1), \ldots (w_m, A_m)\}$ polynomially reduces to \UPos{} (resp. \SUPos{}) for $\langle b^\star \rangle$.
     
     By~\cite{d2ad832a2cc24904b7936daaa5c1c17d}, we know that the characteristic polynomial $p_{\langle b^\star \rangle}(x)$  is the LCM of the characteristic polynomials  $p_{\langle b^i \rangle}(x)$ for $1 \le i \le m$.  If $A_i$ are simple, there are no repeated roots of $p_{\langle b^i \rangle}(x)$.  If this holds for all $i \in \{1, \ldots m\}$, there are no repeated roots of the LCM of $p_{\langle b^1 \rangle}(x), \ldots p_{\langle b^m \rangle}(x)$ as well.  Hence, $p_{\langle b^\star \rangle}(x)$ has no repeated roots.  Similarly, if all eigenvalues of $A_i$ are roots of real algebraic numbers, so are all roots of $p_{\langle b^i \rangle}(x)$.  It follows that all roots of the LCM of $p_{\langle b^1 \rangle}(x), \ldots p_{\langle b^m \rangle}(x)$, i.e. $p_{\langle b^\star \rangle}(x)$, are also roots of real algebraic numbers.
     
     The theorem now follows from the following two known results about LRS.
     \begin{enumerate}
         \item \UPos{} (resp. \SUPos{}) for simple LRS is in $\mathsf{PSPACE}$. Furthermore, if the LRS is of bounded order, \UPos{} (resp. \SUPos{}) is in $\mathsf{PTIME}$~\cite{ouaknine2017ultimate}.
         \item \UPos{} (resp. \SUPos{}) for LRS in which all roots of characteristic polynomial are roots of real algebraic numbers is in $\mathsf{coNP^{PosSLP}}$~\cite{DBLP:conf/stacs/AkshayBMV020}. \qed
    \end{enumerate}
\end{proof}

\subsection{Diagonalizable matrices} 
We now show that \ENSum{} (resp. \EPSum{}) is decidable if each matrix $A_i$ is diagonalizable.  Since diagonalizable matrices strictly subsume simple matrices, this generalizes Theorem~\ref{thm:consqENSumtoUPos}(1). Interestingly, the proof uses a reduction technique that generalizes to properties beyond $\ENSum{}$ and $\EPSum{}$. 

To see an example of a matrix that has 
\begin{wrapfigure}[7]{l}{0.3\textwidth} 
  \begin{minipage}{0.3\textwidth}
    \vspace*{-1cm}
  \begin{equation*}
    A = \begin{bmatrix}
        5 & 12 & -6 \\
        -3 & -10 & 6 \\
        -3 & -12 & 8
    \end{bmatrix}
  \end{equation*}
\end{minipage}
\caption{Diagonalizable matrix}
\label{fig:matrix2}
\end{wrapfigure}
repeated eigenvalues and is diagonalizable, consider the matrix $A$ shown in Fig.~\ref{fig:matrix2}. This has eigenvalues $2, 2$ and $-1$, and is diagonalizable. Specifically, $(-4, 1, 0)$ and $(2, 0, 1)$ are two linearly independent eigenvectors corresponding to eigenvalue $-2$, and $(-1, 1, 1)$ is an eigenvector corresponding to eigenvalue $-1$.  

Let $A$ be a $k \times k$ real diagonalizable matrix.  By definition, there exists an invertible $k \times k$ matrix $S$ and a diagonal $k \times k$ matrix $D$ such that $A = S D S^{-1}$.  Note that although $A \in \mathbb{R}^{k\times k}$,  matrices $S$ and $D$ may have complex entries. 
It is easy to see that for every $i \in \{1, \ldots k\}$, $D[i,i]$ is an eigenvalue of $A$. Moreover, if $\alpha$ is an eigenvalue of $A$ with algebraic multiplicity $\rho$, then  $\alpha$ appears exactly $\rho$ times along the diagonal of $D$. It can also be easily shown that for every $i \in \{1, \ldots k\}$, the $i^{th}$ column of $S$ (resp. $i^{th}$ row of $S^{-1}$) is an eigenvector of $A$ (resp. of $\transp{A}$) corresponding to the eigenvalue $D[i,i]$. Moreover, the columns of $S$ (resp. rows of  $S^{-1}$) form a basis of the vector space $\mathbb{C}^k$.

Let $\alpha_1, \ldots \alpha_m$ be the  eigenvalues of $A$ with algebraic multiplicities $\rho_1, \ldots \rho_m$ respectively.  Without loss of generality, we assume that $\rho_1 \ge \ldots \ge \rho_m$ and that the diagonal of $D$ is partitioned into segments as follows: the first $\rho_1$ entries along the diagonal are $\alpha_1$, the next $\rho_2$ entries are $\alpha_2$, and so on. We refer to these segments as the $\alpha_1$-segment, $\alpha_2$-segment and so on, of the diagonal of $D$. Formally, if $\kappa_i$ denotes $\sum_{j=1}^{i-1} \rho_j$, the $\alpha_i$-segment of the diagonal of $D$ consists of the entries $D[\kappa_i+1, \kappa_i+1], \ldots D[\kappa_i+\rho_i, \kappa_i + \rho_i]$,
all of which are $\alpha_i$.

Since $A$ is a real matrix, its eigenvalues and eigenvectors enjoy special properties. Specifically, the characteristic polynomial of $A$ has all real coefficients.  As a consequence, for every eigenvalue $\alpha$ of $A$ (and hence of $\transp{A}$), its complex conjugate, denoted $\overline{\alpha}$, is also an eigenvalue of $A$ (and hence of $\transp{A}$).  Additionally, the algebraic multiplicities of $\alpha$ and $\overline{\alpha}$ are the same.  This allows us to define a bijection $\bij{D}$ from $\{1, \ldots, k\}$ to $\{1, \ldots k\}$ as follows. If $D[i,i]$ is real, then $\bij{D}(i) = i$.  Otherwise, let $D[i,i] = \alpha \in \mathbb{C}$ and let $D[i,i]$ be the $l^{th}$ element in the $\alpha$-segment of the diagonal of $D$. Then $\bij{D}(i) = j$, where $D[j,j]$ is the $l^{th}$ element in the $\overline{\alpha}$-segment of the diagonal of $D$.  The matrix $A$ being real also implies that for every real eigenvalue $\alpha$ of $A$ (resp. of \transp{A}), there exists a basis of \emph{real eigenvectors} of the corresponding eigenspace. Additionally, for every non-real eigenvalue $\alpha$ and for every set of eigenvectors of $A$ (resp. of  $\transp{A}$) that forms a basis of the eigenspace corresponding to $\alpha$, the component-wise complex conjugates of these basis vectors serve as eigenvectors of $A$ (resp. of $\transp{A}$) and form a basis of the eigenspace corresponding to $\overline{\alpha}$.  

We use the above properties to choose the matrix $S^{-1}$ (and hence $S$) such that $A = S D S^{-1}$.  Specifically, suppose $\alpha$ is an eigenvalue of $A$ (and hence of $\transp{A}$) with algebraic multiplicity $\rho$. Let $\{i+1, \ldots i +\rho\}$ be the set of indices $j$ for which $D[j, j] = \alpha$.  If $\alpha$ is real, the $i+1^{st}, \ldots i+{\rho}^{th}$ rows of $S^{-1}$ are chosen to be real eigenvectors of $\transp{A}$ that form a basis of the eigenspace corresponding to $\alpha$. If $\alpha$ is not real, the $i+1^{th}, \ldots i+{\rho}^{th}$ rows of $S^{-1}$ are chosen to be  (possibly complex) eigenvectors of $\transp{A}$ that form a basis of the eigenspace  corresponding to $\alpha$.  Moreover, using notation introduced above,  the $\bij{D}(i+s)^{th}$ row of $S^{-1}$ is chosen to be the component-wise complex conjugate of the $i+s^{th}$ row of $S^{-1}$, for all $s \in \{1, \ldots \rho\}$.  

\begin{definition}\label{def:simple_perturb}
Let $A = S D S^{-1}$ be a $k\times k$ real diagonalizable matrix. We say that $\EE = (\varepsilon_1, \ldots \varepsilon_k) \in \mathbb{R}^k$  is a \emph{perturbation} w.r.t. $D$ if $\varepsilon_i \neq 0$  and $\varepsilon_i = \varepsilon_{\bij{D}(i)}$ for all $i \in \{1, \ldots k\}$. Furthermore,  the \emph{$\EE$-perturbed variant of $A$} is the matrix $A' = S D' S^{-1}$, where $D'$ is the $k\times k$ diagonal matrix with $D'[i,i] = \varepsilon_i D[i,i]$ for all $i \in \{1, \ldots k\}$.
\end{definition}
In the following, we omit "w.r.t. $D$" and simply say "$\EE$ is a perturbation", when $D$ is clear from the context. Clearly, $A'$ as defined above is a diagonalizable matrix and its eigenvalues are given by the diagonal elements of $D'$. 

\begin{restatable}{lemma}{simplePerturb}
\label{lem:simple-perturb}
    For every real diagonalizable matrix $A = S D S^{-1}$ and  perturbation $\EE ~(\in \mathbb{R}^k)$,  the $\EE$-perturbed variant of $A$ is a real diagonalizable  matrix. Furthermore, if $A$ is a matrix of real algebraic entries and $\EE \in \mathbb{Q}^k$, the $\EE$-perturbed variant of $A$ is a real algebraic diagonablizable matrix.
\end{restatable}

\begin{proof}
    Given a perturbation $\EE$ w.r.t. $D$, we first define $k$ \emph{simple} perturbations $\EE_i ~(1 \le i \le k)$ w.r.t. D as follows: $\EE_i$ has all its components set to $1$, except for the $i^{th}$ component, which is set to $\varepsilon_i$. Furthermore, if $D[i,i]$ is not real, then the $\bij{D}(i)^{th}$ component of $\EE_i$ is also set to $\varepsilon_i$.  It is easy to see from Definition~\ref{def:simple_perturb} that each $\EE_i$ is a perturbation w.r.t. $D$.  Moreover, if $j = \bij{D}(i)$, then $\EE_j = \EE_i$.

    Let $\widehat{\EE} = \{\EE_{i_1}, \ldots \EE_{i_u}\}$ be the set of all \emph{unique}  perturbations w.r.t $D$ among $\EE_1, \ldots \EE_k$. It follows once again from Definition~\ref{def:simple_perturb} that the $\EE$-perturbed variant of $A$ can be obtained as by sequence of $\EE_{i_j}$-perturbations, where $\EE_{i,j} \in \widehat{\EE}$.  Specifically, let $A_{0,\widehat{\EE}} = A$ and $A_{v, \widehat{\EE}}$ be the $\EE_{i_v}$-perturbed variant of $A_{v-1, \widehat{\EE}}$ for all $v \in \{1, \ldots u\}$. Then, the $\EE$-perturbed variant of $A$ is identical to $A_{u, \widehat{\EE}}$. This shows that it suffices to prove the theorem only for simple perturbations $\EE_i$, as defined above.  We focus on this special case below.

    Let $A' = S D' S^{-1}$ be the $\EE_i$-perturbed variant of $A$, and let $D[i,i] = \alpha$. For every $p \in \{1, \ldots k\}$, let $\mathbf{e_p}$ denote the $p$-dimensional unit vector whose $p^{th}$ component is $1$.  Then, $A'\mathbf{e_p}$ gives the $p^{th}$ column of $A'$. We prove the first part of the lemma by showing that $A'~\mathbf{e_p} = (S ~D~' S^{-1})~ \mathbf{e_p} \in \mathbb{R}^{k\times 1}$ for all $p \in \{1, \ldots k\}$.  

    Let $\mathbf{T}$ denote $D' ~S^{-1}~ \mathbf{e_p}$.  Then $\mathbf{T}$ is a column vector with $\mathbf{T}[r] = D'[r,r] ~S^{-1}[r,p]$ for all $r \in \{1, \ldots k\}$.  Let $\mathbf{U}$ denote $S \mathbf{T}$. We know from above that $\mathbf{U}$ is the $p^{th}$ column of the matrix $A'$.  To compute $\mathbf{U}$, recall that the rows of $S^{-1}$ form a basis of $\mathbb{C}^k$. Therefore, for every $q \in \{1, \ldots k\}$, $S^{-1}~\mathbf{e_q}$ can be viewed as transforming the basis of the unit vector $\mathbf{e_q}$ to that given by the rows of $S^{-1}$ (modulo possible scaling by real scalars denoting the lengths of the row vectors of $S^{-1}$). Similarly, computation of $\mathbf{U} = S \mathbf{T}$ can be viewed as applying the inverse basis transformation to $\mathbf{T}$.  It follows that the components of $\mathbf{U}$ can be obtained by computing the dot product of $\mathbf{T}$ and the transformed unit vector $S^{-1}~\mathbf{e_q}$, for each $q \in \{1, \ldots k\}$. In other words, $\mathbf{U}[q] = \mathbf{T}\cdot (S^{-1}~\mathbf{e_q})$.  We show below that each such $\mathbf{U}[q]$ is real.

    By definition, $\mathbf{U}[q] = \sum_{r=1}^k (\mathbf{T}[r]~S^{-1}[r,q])$ $= \sum_{r=1}^k (D'[r,r] ~S^{-1}[r,p] ~S^{-1}[r,q])$.  
    We consider two cases below.    
    \begin{itemize}
        \item If $D[i,i] = \alpha$ is real, recalling the definition of $D'$, the expression for $\mathbf{U}[q]$ simplifies to $\sum_{r=1}^k (D[r,r] ~S^{-1}[r,p] ~S^{-1}[r,q])$ $+$ $(\varepsilon_i - 1)~ \alpha ~S^{-1}[i,p]~ S^{-1}[i,q]$.  Note that $\sum_{r=1}^k (D[r,r] ~S^{-1}[r,p] ~S^{-1}[r,q])$ is the $q^{th}$ component of the vector $(S D S^{-1}) ~\mathbf{e_p} = A ~\mathbf{e_p}$.  Since $A$ is real, so must be the $q^{th}$ component of $A ~\mathbf{e_p}$.  Moreover, since $\alpha$ is real, by our choice of $S^{-1}$, both $S^{-1}[i,p]$ and $S^{-1}[i,q]$ are real.  Since $\varepsilon_i$ is also real, it follows that $(\varepsilon_i - 1)~\alpha ~S^{-1}[i,p] ~S^{-1}[i,q]$ is real.  Hence $\mathbf{U}[q]$ is real for all $q \in \{1, \ldots k\}$.
        \item If $D[i,i] = \alpha$ is not real, from Definition~\ref{def:simple_perturb}, we know that $D'[i,i] = \varepsilon_i~\alpha$ and $D'[\bij{D}(i),\bij{D}(i)] = \varepsilon_i~\overline{\alpha}$. The expression for $\mathbf{U}[q]$ then simplifies to $\sum_{r=1}^k   \big(D[r,r] ~S^{-1}[r,p] ~S^{-1}[r,q]\big)$ $+$ $(\varepsilon_i - 1) ~(\beta + \gamma)$, where $\beta = \alpha ~S^{-1}[i,p] ~S^{-1}[i,q]$ and $\gamma = \overline{\alpha}~ S^{-1}[\bij{D}(i),p]~ S^{-1}[\bij{D}(i),q]$.  By our choice of $S^{-1}$, we know that $S^{-1}[\bij{D}(i),p] = \overline{S^{-1}[i,p]}$ and $S^{-1}[\bij{D}(i),q] = \overline{S^{-1}[i,q]}$.  Therefore, $\beta = \overline{\gamma}$ and hence $(\varepsilon_i - 1)~ (\beta + \gamma)$ is real. By a similar argument as in the previous case, it follows that $\mathbf{U}[q]$ is real for all $q \in \{1, \ldots k\}$.
    \end{itemize}
    The proof of the second part of the lemma follows from the proof of the first part, and from the following facts about real algebraic matrices.
    \begin{itemize}
        \item If $A$ is a real algebraic matrix, then every eigenvalue of $A$ is either a real or complex algebraic number.
        \item If $A$ is diagonalizable, then for every real (resp. complex) algebraic eigenvalue of $A$, there exists a set of real (resp. complex) algebraic eigenvectors that form a basis of the corresponding eigenspace.
    \end{itemize}   

    Now, if $D[i,i] = \alpha$ is real, then $\mathbf{U}[q] = \sum_{r=1}^k (D[r,r] ~S^-1[r,p] ~S^{-1}[r,q])$ $+$ $(\varepsilon_i - 1)~ \alpha ~S^{-1}[i,p]~ S^{-1}[i,q]$.  Notice that $\sum_{r=1}^k (D[r,r] ~S^{-1}[r,p] ~S^{-1}[r,q])$ $=$ $\big(A \mathbf{e_p}\big)[q]$ is a real algebraic number since $A$ is a real algebraic matrix. Furthermore, $\alpha$, $S^{-1}[i,p]$ and $S^{-1}[i,q]$ are all real algebraic numbers.  Since $\varepsilon_i \in \mathbb{Q}$, it follows that $(\varepsilon_i - 1)~ \alpha ~S^{-1}[i,p]~ S^{-1}[i,q]$ is a real algebraic number for all $p, q \in \{1, \ldots k\}$.  Hence,  $\mathbf{U}[q]$ is a real algebraic number.  

    If $D[i,i] = \alpha$ is not real, then we have $\mathbf{U}[q] = \sum_{r=1}^k   \big(D[r,r] ~S^{-1}[r,p] ~S^{-1}[r,q]\big)$ $+$ $(\varepsilon_i - 1) ~(\beta + \gamma)$, where $\beta$ and $\gamma$ are as defined above.  Since $\alpha$, $S^{-1}[i,p]$, $S^{-1}[i,q]$, $S^{-1}[\bij{D}(i), p]$ and $S^{-1}[\bij{D}(i), q]$ are all algebraic numbers, and $\varepsilon_i$ is rational, it follows that $(\varepsilon_i - 1) ~(\beta + \gamma)$ $=$ $(\varepsilon_i - 1) \times 2 \times \myRe(\alpha ~S^{-1}[i,p] ~S^{-1}[i,q])$ is a real algebraic number.  Hence,  $\mathbf{U}[q]$ is a real algebraic number.

\end{proof}

\noindent Recall that the diagonal of $D$ is partitioned into $\alpha_i$-segments, where each $\alpha_i$ is an eigenvalue of $A = SDS^{-1}$ with algebraic multiplicity $\rho_i$. We now use a similar idea to segment a perturbation $\EE$  w.r.t. $D$. Specifically, the first $\rho_1$ elements of $\EE$ constitute the $\alpha_1$-segment of $\EE$, the next $\rho_2$ elements of $\EE$ constitute the $\alpha_2$-segment of $\EE$ and so on.  The following definition formalizes this.

\begin{definition}\label{def:seg-perturb}
    A perturbation $\EE = (\varepsilon_1, \ldots \varepsilon_k)$ is said to be \emph{segmented} if the $j^{th}$ element (whenever present) in every segment of $\EE$ has the same value, for all $1\le j \le \rho_1$.   Formally, if $i = \sum_{s=1}^{l-1} \rho_s + j$ and $1 \le j \le \rho_l \le \rho_1$, then $\varepsilon_i = \varepsilon_j$.  
\end{definition}

Clearly, the first $\rho_1$ elements of a segmented perturbation $\EE$ define the whole of $\EE$. As an example, suppose $(\alpha_1, \alpha_1, \alpha_1, \alpha_2, \alpha_2, \overline{\alpha_2}, \overline{\alpha_2}, \alpha_3)$ is the diagonal of $D$, where $\alpha_1, \alpha_2, \overline{\alpha_2}$ and $\alpha_3$ are distinct eigenvalues of $A$. There are four segments of the diagonal of $D$ (and of $\EE$) of lengths $3, 2, 2$ and $1$ respectively. Example segmented perturbations in this case are $(\varepsilon_1, \varepsilon_2, \varepsilon_3, \varepsilon_1, \varepsilon_2, \varepsilon_1, \varepsilon_2, \varepsilon_1)$ and $(\varepsilon_3, \varepsilon_1,$ $\varepsilon_2, \varepsilon_3, \varepsilon_1, \varepsilon_3, \varepsilon_1, \varepsilon_3)$.  If $\varepsilon_1 \neq \varepsilon_2$, a perturbation that is \emph{not segmented} is $\widetilde{\EE} = (\varepsilon_1, \varepsilon_2, \varepsilon_3, \varepsilon_2, \varepsilon_3, \varepsilon_2, \varepsilon_3, \varepsilon_1)$. Notice that the first element of the second segment of $\widetilde{\EE}$ is $\varepsilon_2$, while the first element of the first segment is $\varepsilon_1$.

\begin{definition}\label{def:rotation}
    Given a segmented perturbation $\EE = (\varepsilon_1, \ldots \varepsilon_k)$ w.r.t. $D$, a \emph{rotation} of $\EE$, denoted $\tau_{D}(\EE)$, is the segmented perturbation $\EE' = (\varepsilon_1', \ldots \varepsilon_k')$ in which $\varepsilon_{(i \mod \rho_1)+1}' = \varepsilon_i$ for $i \in \{1, \ldots \rho_1\}$, and all other $\varepsilon_i'$s are defined as in Definition~\ref{def:seg-perturb}. 
\end{definition}

Continuing with our example, if $\EE = (\varepsilon_1, \varepsilon_2, \varepsilon_3, \varepsilon_1, \varepsilon_2, \varepsilon_1, \varepsilon_2, \varepsilon_1)$ is a segmented perturbation w.r.t. $D$, then $\tau_{D}(\EE) =$ $(\varepsilon_3, \varepsilon_1, \varepsilon_2, \varepsilon_3, \varepsilon_1, \varepsilon_3, \varepsilon_1, \varepsilon_3)$, $\tau_{D}^2(\EE)$ $=$ $(\varepsilon_2, \varepsilon_3, \varepsilon_1,$ $\varepsilon_2, \varepsilon_3, \varepsilon_2, \varepsilon_3, \varepsilon_2)$ and $\tau_{D}^3(\EE) = \EE$.

\begin{lemma}\label{lem:scale-diagonalizable}
Let $A = S D S^{-1}$ be a $k\times k$ real diagonalizable matrix with eigenvalues $\alpha_i$ of algebraic multiplicity $\rho_i$.  Let $\EE = (\varepsilon_1, \ldots \varepsilon_k)$ be a segmented perturbation w.r.t. $D$ such that all $\varepsilon_j$s have the same sign, and let $A_u$ denote the $\tau_{D}^{u}(\EE)$-perturbed variant of $A$ for $0 \le u < \rho_1$, where $\tau^0(\EE) = \EE$. Then $A^n ~=~ \frac{1}{\big(\sum_{j=1}^{\rho_1} \varepsilon_j^n\big) }\sum_{u=0}^{\rho_1-1} A_u^n$, for all $n \ge 1$. 
\end{lemma}
\begin{proof}
Let $\EE_u$ denote $\tau^u(\EE)$ for $0 \le u < \rho_1$, and let $\EE_u[i]$ denote the $i^{th}$ element of $\EE_u$ for $1 \le i \le k$.  It follows from Definitions~\ref{def:seg-perturb} and \ref{def:rotation} that for each $i, j \in \{1, \ldots \rho_1\}$, there is a unique $u \in \{0, \ldots \rho_1-1\}$ such that $\EE_u[i] = \varepsilon_j$. Specifically, $u = i-j$ if $i \ge j$, and $u = (\rho_1-j)+i$ if $i < j$. 
Furthermore, Definition~\ref{def:seg-perturb} ensures that the above property holds not only for $i \in \{1, \ldots \rho_1\}$, but for all $i \in \{1, \ldots k\}$.  

Let $D_u$ denote the diagonal matrix with $D_u[i,i] = \EE_u[i] D[i,i]$ for $0 \le i < \rho_1$.  Then $D_u^n$ is the diagonal matrix with $D_u^n[i,i] = \big(\EE_u[i] D[i,i]\big)^n$ for all $n \ge 1$. It follows from the definition of $A_u$  that $A_u^n = S~ D_u^n ~S^{-1}$ for  $0 \le u < \rho$ and $n \ge 1$.  Therefore, $\sum_{u=0}^{\rho_1-1} A_u^n$ $=$ $S ~\big(\sum_{u=0}^{\rho_1-1} D_u^n\big)~S^{-1}$.  Now, $\sum_{u=0}^{\rho_1-1} D_u^n$ is a diagonal matrix whose $i^{th}$ element along the diagonal is $\sum_{u=0}^{\rho_1-1} \big(\EE_u[i]  D[i,i]\big)^n$ $=$ $\big(\sum_{u=0}^{\rho_1-1} \EE_u^n[i]\big)~D^n[i,i]$. By virtue of the property mentioned in the previous paragraph, $\sum_{u=0}^{\rho_1-1} \EE_u^n[i] = \sum_{j=1}^{\rho_1} \varepsilon_j^n$ for $1 \le i \le k$. Therefore,  $\sum_{u=0}^{\rho_1-1} D_u^n$ $=$ $\big(\sum_{j=1}^{\rho_1} \varepsilon_j^n\big) ~D^n$, and hence, $\sum_{u=0}^{\rho_1-1} A_u^n$ $=$ $\big(\sum_{j=1}^{\rho_1} \varepsilon_j^n\big)~ S~D^n~S^{-1}$ $=$ $\big(\sum_{j=1}^{\rho_1} \varepsilon_j^n\big)~ A^n$.  Since all $\varepsilon_j$s have the same sign and are non-zero, $\big(\sum_{j=1}^{\rho_1} \varepsilon_j^n\big)$ is non-zero for all $n \ge 1$. It follows that $A^n ~=~ \frac{1}{\big(\sum_{j=1}^{\rho_1} \varepsilon_j^n\big) }\sum_{u=0}^{\rho_1-1} A_u^n$. \qed
\end{proof}

A \emph{positive scaling invariant} property $\mathcal{P}$ of the
sequence of matrices $B_n = \sum_{i=1}^m w_i A_i^n$ is one that is
invariant under scaling of the $A_i$s by a positive
real. That is, if we scale all $A_i$s by the same positive
real, whether $\mathcal{P}$ holds or not stays unchanged. 
Examples of such properties include $\ENSum{}$, $\EPSum{}$, 
non-negativity (resp. positivity), i.e. is $B_n[i,j] \ge 0$ (resp. $< 0$) for all $n\ge 1$, $1 \le i,j \le k$, existence of zero (i.e. is $B_n$ the all $0$-matrix for some $n \ge 1$), existence of a zero element (i.e. is $B_n[i,j] = 0$ for some $n \ge 1$ and some $i, j \in \{1, \ldots k\})$, variants of the $r-$non-negativity (resp. $r-$positivity and $r-$zero) problem (i.e Does there exist at least/exactly/at most $r$ non-negative (resp. positive/zero) elements in $B_n$ (for all $n \geq 1$) for a given $r \in [1,k]$)  etc.

\noindent We now state and prove the main theorem of this section, which essentially gives a reduction technique to reduce decision problems for a \emph{positive scaling invariant} property on sets of diagonalizable matrices to a decision problem for the same property of simple matrices : 

\begin{theorem}
\label{thm:diag-to-simple-technique}
    The decision problem for every positive scaling invariant property on rational diagonalizable matrices effectively reduces to the decision problem for the property on real algebraic simple matrices.
\end{theorem}

\begin{proof}
    The proof uses a variation of the idea used in the proof of Lemma~\ref{lem:scale-diagonalizable}. We assume that each matrix $A_i$ is in $\mathbb{Q}^{k \times k}$ and $A_i = S_i D_i S_i^{-1}$, where $D_i$ is a diagonal matrix with segments along the diagonal arranged in descending order of algebraic multiplicities of the corresponding eigenvalues.  Let $\nu_i$ be the number of distinct eigenvalues of $A_i$, and let these eigenvalues be $\alpha_{i,1}, \ldots \alpha_{i, \nu_i}$.  Let $\mu_i$ be the largest algebraic multiplicity among the eigenvalues of $A_i$ and let $\mu = lcm(\mu_1, \ldots \mu_m)$. We now choose \emph{positive} rationals $\varepsilon_1, \ldots \varepsilon_\mu$ such that (i) all $\varepsilon_j$s are distinct, and (ii) for every $i \in \{1, \ldots m\}$, for every distinct $j, l \in \{1, \ldots \nu_i\}$ and for every distinct $p, q \in \{1, \ldots \mu\}$, we have $\frac{\varepsilon_p}{\varepsilon_q} \neq \lvert \frac{\alpha_{i,j}}{\alpha{i,l}} \rvert$.  Since $\mathbb{Q}$ is a dense set, such a choice of $\varepsilon_1, \ldots \varepsilon_\mu$ can always be made once all $\lvert \frac{\alpha_{i,j}}{\alpha_{i,l}}\rvert$s are known, even if within finite precision bounds. 

    For $1 \le i \le m$, let $\eta_i$ denote $\mu/\mu_i$.  We now define $\eta_i$ distinct and segmented perturbations w.r.t. $D_i$ as follows, and denote these as $\EE_{i,1}, \ldots \EE_{i, \eta_i}$.  For $1 \le j \le \eta_i$, the first $\mu_i$ elements (i.e. the first segment) of $\EE_{i,j}$ are $\varepsilon_{(j-1)\mu_i + 1}, \ldots \varepsilon_{j\mu_i}$ (as chosen in the previous paragraph), and all other elements of $\EE_{i,j}$ are defined as in Definition~\ref{def:seg-perturb}. For each $\EE_{i,j}$ thus obtained, we also consider its rotations $\tau_{D_i}^u(\EE_{i,j})$ for $0 \le u < \mu_i$.  For $1 \le j \le \eta_i$ and $0 \le u < \mu_i$, let $A_{i,j,u} = S_i ~D_{i,j,u}~S_i^{-1}$ denote the $\tau_{D_i}^u(\EE_{i,j})$-perturbed variant of $A_i$.  It follows from Definition~\ref{def:simple_perturb} that if we consider the set of diagonal matrices $\{D_{i,j,u} \mid 1 \le j \le \eta_i$, $0 \le u < \mu_i\}$, then for every $p \in \{1, \ldots k\}$ and for every $q \in \{1, \ldots \mu\}$, there is a unique $u$ and $j$ such that $D_{i,j,u}[p,p] = \varepsilon_q$. Specifically, $j = \lfloor q/\mu_i\rfloor$. To find $u$, let $\EE_{i,j}[p]$ be the $\widehat{p}^{th}$ element in a segment of $\EE_{i,j}$, where $1 \le \widehat{p} \le \mu_i$, and let $\widehat{q}$ be $q\mod \mu_i$. Then, $u = (\widehat{p} - \widehat{q})$ if $\widehat{p} \ge \widehat{q}$ and $u = (\mu_i - \widehat{q}) + \widehat{p}$ otherwise. By our choice of $\varepsilon_j$s, we also know from Lemma~\ref{lem:simple-perturb} that all $A_{i, j, u}$s are rational simple matrices.

    Using the same reasoning as in Lemma~\ref{lem:scale-diagonalizable}, we can now show that $A_i^n = \frac{1}{\big(\sum_{j=1}^\mu \varepsilon_j^n\big)}\times \big(\sum_{j=1}^{\eta_i} \sum_{u=0}^{\mu_i-1} A_{i,j,u}^n\big)$.  It follows that $\sum_{i=1}^m w_i A_i^n$ $=$ $\frac{1}{\big(\sum_{j=1}^\mu \varepsilon_j^n\big)}\times \big(\sum_{i=1}^m\sum_{j=1}^{\eta_i} \sum_{u=0}^{\mu_i-1} w_i A_{i,j,u}^n\big)$. Since all $\varepsilon_j$s are positive reals, $\sum_{j=1}^\mu \varepsilon_j^n$ is a positive real for all $n \ge 1$. 

    Hence, for each $p, q \in \{1, \ldots k\}$,  $\sum_{i=1}^m w_i A_i^n[p,q]$ is $> 0$, $< 0$ or $= 0$ if and only if $\big(\sum_{i=1}^m\sum_{j=1}^{\eta_i}\sum_{u=0}^{\mu_i-1} w_i A_{i,j,u}^n[p,q]\big)$ is $> 0$, $< 0$ or $= 0$, respectively.  Since each $A_i$ is a rational matrix and each $\varepsilon_t$ is a rational, by Lemma~\ref{lem:simple-perturb}, we know that $A_{i,j,u}$ is a real algebraic simple matrix. \qed

\end{proof}

The reduction in the proof of Theorem~\ref{thm:diag-to-simple-technique} can be encoded as an algorithm, as shown in Algorithm~\ref{algo:perturb}.

\begin{algorithm}
\caption{Reduction procedure for diagonalizable matrices}
\label{algo:perturb}
\begin{algorithmic}[1]

\Statex {\bfseries Input:~} $\AAA = \{(w_i, A_i)~:~1 \le i \le m,~ w_i \in \mathbb{Q}, ~A_i \in \mathbb{Q}^{k\times k}$ and diagonalizable$\}$
\Statex {\bfseries Output:~} $\BBB = \{(v_i, B_i): ~1 \le i \le t, ~v_i \in \mathbb{Q}, ~B_i  \in \mathbb{RA}^{k\times k}$ are simple$\}$
\Statex \hspace*{1.5cm} s.t. $\left(\sum_{i=1}^m w_i A_i^n\right)$ $ = f(n)\left(\sum_{i=1}^t v_i B_i^n\right)$, where $f(n) > 0$ for all $n \ge 0$?

\State $P \gets \{1\}$; \Comment{Initialize set of forbidden ratios of various $\varepsilon_j$s}
\For{$i$ in $1$ through $m$}\Comment{For each matrix $A_i$}
    \State $R_i \gets \{(\alpha_{i,j}, \rho_{i,j}) ~:~ \alpha_{i,j}$ is eigenvalue of $A_i$ with algebraic multiplicity $\rho_{i,j}\}$;
    \State $D_i \gets$ Diagonal matrix of $\alpha_{i,j}$-segments ordered in decreasing order of $\rho_{i,j}$;
    \State $S_i \gets$ Matrix of linearly independent eigenvectors of $A_i$ s.t. $A_i = S_i D_i S_i^{-1}$;
    \State $P \gets P ~\cup$ $\big\{\left| {\alpha_{i,j}}/{\alpha_{i,l}} \right| ~:~ \alpha_{i,j}, \alpha_{i,l}$ are eigenvalues in $R_i\big\}$;~~
    $\mu_i  \gets \max_j \rho_{i,j}$
\EndFor
\State $\mu = lcm(\mu_1, \ldots \mu_m)$; \Comment{Count of $\varepsilon_j$s needed} 
\For{$j$ in $1$ through $\mu$}\Comment{Generate all required $\varepsilon_j$s}
 \State Choose $\varepsilon_j \in \mathbb{Q}$ s.t. $\varepsilon_j > 0$ and $\varepsilon_j \not\in \{\pi \varepsilon_p ~:~ 1 \le p < j, ~\pi \in P\}$;
\EndFor
\State $\BBB \gets \emptyset$;\Comment{Initialize set of (weight, simple matrix) pairs}
\For{$i$ in $1$ through $m$}\Comment{For each matrix $A_i$}
 \State $\nu_i \gets \mu/\mu_i$; \Comment{Count of segmented perturbations to be rotated for $A_i$}
 \For{$j$ in $0$ through $\nu_i - 1$} \Comment{For each segmented perturbation}
  \State $\EE_{i,j} \gets$ Seg. perturbn. w.r.t. $D_i$ with first $\mu_i$ elements being $\varepsilon_{j\mu_i+1}, \ldots \varepsilon_{(j+1)\mu_i}$;
  \For{$u$ in $0$ through $\mu_i - 1$} \Comment{For each rotation of $\EE_{i,j}$}
     \State $A_{i,j,u} \gets \tau_{D_i}^u(\EE_{i,j})$-perturbed variant of $A$;
     \State $\BBB \gets \BBB \cup \{(w_i, A_{i,j,u})\}$; \Comment{Update $\AAA'$}
  \EndFor
 \EndFor
\EndFor

\State \Return $\BBB$;
\end{algorithmic}
\end{algorithm}

\medskip

\noindent Theorem ~\ref{thm:diag-to-simple-technique} has several interesting consequences, of which we focus on two here.

\begin{corollary}\label{cor:others}
    Given $\AAA = \{(w_1, A_1), \ldots (w_m, A_m)\}$, where each $w_i \in \mathbb{Q}$ and $A_i \in \mathbb{Q}^{k\times k}$ is diagonalizable, and a real value $\varepsilon > 0$, there exists $\BBB = \{(v_1, B_1),$ $\ldots (v_M, B_M)\}$, where each $v_i \in \mathbb{Q}$ and each $B_i \in \mathbb{RA}^{k\times k}$ is simple, such that \\ $ \left| \sum_{i=0}^m w_i A_i^n[p,q] - \sum_{j=0}^{M} v_j B_j^n[p,q]\right| < \varepsilon^n$ for all $p, q \in \{1, \ldots k\}$ and all $n \ge 1$. 
\end{corollary}

\begin{proof}
    Let $\alpha$ be the eigenvalue with largest modulus among all eigenvalues of $A_1, \ldots A_m$.  Let $\gamma$ be $\max_{i=1}^m \max_{p,q,r=1}^k \left|S_i[p,r] ~S_i^{-1}[r,q] \right|$.  The proof now follows from the proof of Theorem~\ref{thm:main} by choosing $\varepsilon_1 = 1$ and all other $\varepsilon_j$s such that $\varepsilon_j < \min\left(1, \frac{\varepsilon}{\mu k \cdot \left|\alpha\right|\cdot \max(1,\gamma)}\right)$. \qed
\end{proof}
 
 \noindent Using Theorem~\ref{thm:diag-to-simple-technique}, we can now show the decidability of \ENSum{} (resp. \EPSum{}) for sets of diagonalizable matrices.    
 
\begin{corollary}\label{thm:main}
    \ENSum{} and \EPSum{} are decidable for $\AAA = \{(w_1, A_1), \ldots (w_m, A_m)\}$ if all $A_i$s are rational diagonalizable matrices and all $w_i$s are rational.  
\end{corollary}

\begin{proof}
    Using the reduction technique used in the proof of Theorem~\ref{thm:diag-to-simple-technique}, we can effectively reduce \ENSum{} (resp. \EPSum{}) for $\AAA = \{(w_1, A_1), \ldots (w_m, A_m)\}$ to \ENSum{} (resp. \EPSum{}) for $\AAA' = \bigcup_{i=1}^m \bigcup_{j=1}^{\eta_i}\bigcup_{u=0}^{\mu_i-1} \{(w_i, A_{i,j,u})\}$, where each $A_{i,j,u}$ is a matrix of real algebraic entries. Now we appeal to part (1) of Theorem \ref{thm:consqENSumtoUPos}. In fact, the proof of part (1) of Theorem~\ref{thm:consqENSumtoUPos} relies on a technically intricate technique due to Ouaknine and Worrell~\cite{ouaknine2017ultimate} that establishes the decidability and complexity (upper) bound of \UPos{} and \SUPos{} for simple LRS with rational coefficients and rational initial values. Interestingly, the same proof technique also yields the decidability of \UPos{} and \SUPos{} of simple LRS with real algebraic numbers as coefficients and initial values, which finally completes our proof 
    \qed
\end{proof}

\paragraph{Remark :} The reduction technique (Theorem~\ref{thm:diag-to-simple-technique} from checking a \emph{positive scaling invariant} property for a set of weighted diagonalizable matrices to checking the same property for a set of weighted simple matrices) makes no assumptions about the inner working of the decision procedure for simple matrices. Given \emph{any} black-box decision procedure for checking \emph{any} such property for a set of weighted simple matrices, our reduction tells us how a corresponding decision procedure for checking the same property for a set of weighted diagonalizable matrices can be constructed. 

Diagonalizable matrices have an exponential form solution with constant coefficients for the exponential terms. Hence an algorithm that exploits this specific property of the exponential form (like Ouaknine and Worrell's algorithm~\cite{ouaknine2017ultimate}, originally proposed for checking ultimate positivity of simple LRS) can be used directly in the case of diagonalizable matrices. However, the proposed reduction technique in this paper is neither specific to this algorithm nor does it rely on any special property the exponential form of the solution.

\section{Conclusion}
\label{sec:conc}
\vspace*{-0.2cm}    
In this paper, we investigated eventual non-negativity  and positivity for matrices and the weighted sum of powers of matrices (\ENSum{} / \EPSum{}). We showed their links to problems on linear recurrence sequence showed both lower and upper bounds.  We developed a new perturbation-based technique that allowed us to decide the \ENSum{}, \EPSum{}, and \ENId{} problems for the class of diagonalizable matrices. It is interesting to note that the class of LRS whose companion matrices are diagonalizable coincides with the class of simple LRS, but the class of diagonalizable matrices in general is strictly larger than the class of simple matrices. This shows a marked difference between the corresponding problems on matrices and on LRS.

In this work, we considered matrices with rational entries and weights. However, most of our results hold even with real-algebraic matrices. Of course, this would require adapting the complexity notions and would depend on corresponding results for ultimate positivity for linear recurrences and related problems over reals.
 Also, while we have focused on \emph{eventuality} problems, it still needs to be seen whether our techniques can be adapted  for other problems of interest like the \emph{existence} of a matrix power where all entries are non-negative or zero. Finally, the line of work started in this paper could lead to effective algorithms and applications in varied areas ranging from control theory systems to cyber-physical systems, where eventual properties of matrices play a crucial role.

\bibliographystyle{splncs04}
\bibliography{references}

\end{document}

%% file: macro.tex
\newcommand{\UPos}{\ensuremath{\mathsf{UNN_{LRS} }}}
\newcommand{\myVec}[1]{\ensuremath{\mathbf{{#1} }}}
\newcommand{\SUPos}{\ensuremath{\mathsf{UP_{LRS} }}}
\newcommand{\ENId}{\ensuremath{\mathsf{ENN_{Mat}}}}
\newcommand{\EPosId}{\ensuremath{\mathsf{EP_{Mat} }}}
\newcommand{\ENSum}{\ensuremath{\mathsf{ENN_{SoM}}}}
\newcommand{\EPSum}{\ensuremath{\mathsf{EP_{SoM} }}}

\newcommand{\AAA}{\ensuremath{\mathfrak A}}
\newcommand{\BBB}{\ensuremath{\mathfrak B}}
\newcommand{\msum}[2]{\ensuremath{\sum {#1}^{{#2}}}}
\newcommand{\mycoNP}{\ensuremath{\mathsf{coNP}}}
\newcommand{\myP}{\ensuremath{\mathsf{P}}}
\newcommand{\myNP}{\ensuremath{\mathsf{NP}}}
\newcommand{\myPTIME}{\ensuremath{\mathsf{PTIME}}}
\newcommand{\myPSPACE}{\ensuremath{\mathsf{PSPACE}}}

\newcommand{\transp}[1]{\ensuremath{{#1}^{\mathsf T}}}
\newcommand{\bij}[1]{\ensuremath{{h}_{{#1}}}}
\newcommand{\myRe}{\ensuremath {\mathsf Re}}

\newcommand{\EE}{\ensuremath{\mathcal{E}}}